\def\fullversion{1}

\ifnum\fullversion=0
\documentclass[conference]{IEEEtran}
\else
\documentclass[11pt]{article}
\usepackage{fullpage}
\usepackage{authblk}

\fi

\usepackage{hyperref} 
\hypersetup{
    colorlinks,
    linkcolor={red!50!black},
    citecolor={blue!50!black},
    urlcolor={blue!80!black}
}

\usepackage[
	lambda,
	operators,
	advantage,
	sets,
	adversary,
	landau,
	probability,
	notions,
	logic,
	ff,
	mm,
	primitives,
	events,
	complexity,
	asymptotics,
	oracles,
	keys]{cryptocode}
\pcfixcleveref
\usepackage[T1]{fontenc}
\usepackage{graphicx}
\usepackage[utf8]{inputenc}

\usepackage{tikz}
\usetikzlibrary{backgrounds,arrows.meta,calc,positioning,fit}

\usepackage{pgfplots} 
\usepgfplotslibrary{groupplots}
\pgfplotsset{compat=1.18}

\usepackage{booktabs}
\usepackage{multirow}

\usepackage{csquotes} 
\usepackage[open, depth=2, openlevel=1, numbered]{bookmark} 
\usepackage[dvipsnames,table]{xcolor}
\definecolor{apricot}{cmyk}{0,0.32,0.52,0}
\usepackage{xspace}
\usepackage{pifont} 
\usepackage{varwidth} 

\usepackage[capitalize,noabbrev]{cleveref}

\usepackage[most]{tcolorbox} 

\usepackage[inline]{enumitem}

\setlist[description]{noitemsep, leftmargin = 1em}

\setlist[itemize]{leftmargin=1.5em,itemsep=0pt}

\newlist{inlineenum}{enumerate*}{1}
\setlist[inlineenum]{label=(\roman*),ref=(\roman*)}
\newlist{inlineenum*}{enumerate*}{1}
\setlist[inlineenum*]{label={},afterlabel={}}

\newlist{constraintlist}{enumerate}{1}
\setlist[constraintlist]{label=\textbf{(C\arabic*)}, ref=C\arabic*, leftmargin=2.4em, itemsep=3pt}
\crefname{constraintlisti}{constraint}{constraints}
\Crefname{constraintlisti}{Constraint}{Constraints}

\usepackage{amssymb}
\usepackage{amsmath}
\usepackage{bbm}

\usepackage{amsthm}
\newtheorem{theorem}{Theorem}
\newtheorem{lemma}[theorem]{Lemma}

\newtheorem{corollary}[theorem]{Corollary}

\theoremstyle{definition}

\theoremstyle{plain}

\tcbuselibrary{breakable}
\newtcolorbox{textdef}{
  colback=white,
  colframe=black,
  boxrule=0.5pt,
  arc=0pt,
  left=1ex,
  right=1ex,
  top=1ex,
  bottom=1ex
}

\usepackage{orcidlink} 
\ifnum\fullversion=1
\usepackage[font=small]{caption} 
\fi
\newcommand{\Iceberg}{\textsf{Iceberg}\xspace}
\DeclareRobustCommand{\scheme}[1]{%
  \ensuremath{\textsf{Iceberg}\ifnum#1=0 _{\mathrm{base}}\fi}%
}

\newcommand{\defeq}{\gets}
\newcommand{\mathsc}[1]{{\normalfont\textsc{#1}}}
\newcommand{\vpss}[1]{\mathsf{VPSS}_{#1}}
\newcommand{\interval}[1]{\left[#1\right]}
\newcommand{\verts}[1]{\left\vert #1\right\vert}
\newcommand{\grgen}{\mathsf{GrGen}}
\newcommand{\setup}{\mathsf{Setup}}
\newcommand{\keygen}{\mathsf{KeyGen}}
\newcommand{\gen}{\mathsf{Gen}}
\newcommand{\ver}{\mathsf{Verify}}
\newcommand{\vpssver}[1]{\vpss{#1}.\ver} 
\newcommand{\agg}{\mathsf{Agg}}
\newcommand{\recover}{\mathsf{Recover}}
\newcommand{\preround}{\mathsf{PreRound}}
\newcommand{\preagg}{\mathsf{PreAgg}}
\newcommand{\signround}{\mathsf{SignRound}}
\newcommand{\signagg}{\mathsf{SignAgg}}
\newcommand{\signaggext}{\mathsf{SignAggExt}}
\newcommand{\signo}{\mathsc{OSignRound}}
\newcommand{\preroundo}{\mathsc{OPreRound}}
\newcommand{\keygeno}{\mathsc{OKeyGen}}
\newcommand{\keyaggcoef}{\mathsf{KeyAggCoef}}
\newcommand{\musigcoef}{\mathsf{KeyAggCoef}}
\newcommand{\keyagg}{\mathsf{KeyAgg}}
\newcommand{\Hprf}{\mathsf{H}_{\mathrm{prf}}}
\newcommand{\Hsig}{\mathsf{H}_{\mathrm{sig}}}
\newcommand{\Hagg}{\mathsf{H}_{\mathrm{agg}}}
\newcommand{\Hnon}{\mathsf{H}_{\mathrm{non}}}
\newcommand{\Hnonover}{\mathsf{H}_{\overline{\mathrm{non}}}}

\newcommand{\ctrs}{\mathit{ctrs}}
\newcommand{\param}{\mathit{par}} 
\newcommand{\gparam}{(\GG,p,g)}
\newcommand{\ssum}{\textstyle \sum}
\newcommand{\sprod}{\textstyle \prod}
\newcommand{\sbinom}[2]{\textstyle \binom{#1}{#2}}
\newcommand{\Zpp}{\mathbb{Z}_p}
\newcommand{\bool}{\{0,1\}}
\newcommand{\str}{\bool^*}
\newcommand{\pcsc}{\,;~}
\newcommand{\tX}{\widetilde{X}}

\newtcbox{\hlbox}{on line,
  colback=gray!10,
  left=0pt,right=0pt,top=0pt,bottom=0pt,
  boxrule=0pt,
  enhanced,
  tcbox raise base,
  math upper
}
\newcommand{\highlight}[1]{\hlbox{#1}}
\newcommand{\tr}{\mathtt{true}}

\newcommand{\msg}{\mathit{out}}
\renewcommand{\state}{\mathit{state}}

\ifdefined\pcgame 
\renewcommandx{\pcgame}[4][3=\adv,4=(\secpar)]{{\operatorname{#1}_{#2}^{#3}#4}}
\else
\newcommandx{\pcgame}[4][3=\adv,4=(\secpar)]{{\operatorname{#1}_{#2}^{#3}#4}}
\fi
\unless\ifdefined\pclinecomment
\newcommand{\pclinecomment}[2][0em]{\hspace{#1}{\mbox{/\!\!/ } \text{\scriptsize#2}}}
\fi

\newcommand{\aomdl}{\ensuremath{\mathsf{AOMDL}}\xspace}
\newcommand{\DL}{\mathsc{DLog}}
\newcommand{\tsmseufcma}{\ensuremath{\mathsf{TS\textnormal{-}MS\textnormal{-}EUF\textnormal{-}CMA}}\xspace}
\newcommand{\corrupt}{\text{corrupt}}
\newcommand{\honest}{\text{honest}}
\newcommand{\thres}{\ensuremath{t}}      
\newcommand{\groupsize}{\ensuremath{n}}  

\newcommand{\rot}[1]{\rotatebox[origin=l]{90}{#1}}
\newcommand{\cmark}{\textcolor{green!55!black}{\ding{51}}}   
\newcommand{\xmark}{\textcolor{red!65!black}{\ding{55}}}     
\newcommand{\nax}{\textcolor{black!45}{--}}                  

\newcommand{\defitem}[1]{\underline{#1}}
\newcommand{\mypar}[1]{\paragraph{#1}}

\newcommand{\nestedmusig}{\ensuremath{\mathsf{NestedMuSig2}}\xspace} 
\newcommand{\sid}{\ensuremath{\mathit{sid}}}                        
\newcommand{\sidspace}{\ensuremath{\mathcal{SID}}}                  
\newcommand{\ctx}{\ensuremath{\mathit{ctx}}}                        
\newcommand{\aset}{\ensuremath{\mathbf{a}}}
\newcommand{\asets}{\ensuremath{\mathbf{A}}}
\newcommand{\asetsof}[1]{\ensuremath{\asets_{#1}}}

\crefname{paragraph}{Section}{Sections}
\Crefname{paragraph}{Section}{Sections}

\newcommand{\papertitle}{Enabling Threshold Custody for the Lightning Network with Nested Threshold Multi-Signatures}
\newcommand{\paperkeywords}{Lightning Network, Threshold Custody, Nested
Threshold Multi-Signatures, MuSig2, Schnorr Signatures}
\ifnum\fullversion=0
  \newcommand{\paperauthors}{Anonymous Submission}
\else
  \newcommand{\paperauthors}{Paul Gerhart, Nadav Kohen, Jesse Posner, Matias Furszyfer}
\fi
\hypersetup{
    pdftitle={\papertitle},
    pdfauthor={\paperauthors},
    pdfkeywords={\paperkeywords},
}

\makeatletter
\ifnum\fullversion=0
  \newcommand{\tabcaption}[2]{\caption{#1}\label{#2}}
  \newcommand{\tabcaptionend}{}
\else
  \newcommand{\tabcaption}[2]{\gdef\ic@cap{#1}\gdef\ic@lab{#2}}
  \newcommand{\tabcaptionend}{\caption{\ic@cap}\label{\ic@lab}}
  \let\ic@paragraph\paragraph
  \renewcommand{\paragraph}[1]{\ic@paragraph{#1.}}
\fi
\makeatother

\ifnum\fullversion=0
  \newcommand{\prototypecite}{\cite{prototype}}
\else
  \newcommand{\prototypecite}{\cite{prototypefull}}
\fi

\ifnum\fullversion=0
  \newcommand{\fullscale}{1}
  \newcommand{\tabscale}{1}
  \newcommand{\figwidth}{0.66\linewidth}
  \newcommand{\gamescale}{1}
  \newcommand{\tabfont}{\footnotesize}
\else
  \newcommand{\figwidth}{6.2cm}
  \newcommand{\fullscale}{0.9}
  \newcommand{\gamescale}{0.75}
  \newcommand{\tabfont}{\small}
  \newcommand{\tabscale}{0.45}
\fi

\newcommand{\shrinkfit}[2]{%
  \sbox\icfitbox{#2}%
  \ifdim\wd\icfitbox>\columnwidth
    \resizebox{#1\columnwidth}{!}{\usebox\icfitbox}%
  \else
    \resizebox{#1\wd\icfitbox}{!}{\usebox\icfitbox}%
  \fi}

\ifnum\fullversion=0
  \newenvironment{wfigure}[1][tp]{\begin{figure*}[#1]}{\end{figure*}}
\else
  \newenvironment{wfigure}[1][tp]{\begin{figure}[#1]}{\end{figure}}
\fi

\ifnum\fullversion=0
  \newcommand{\pairfigstart}{\begin{figure*}[tp]\centering
    \begin{minipage}[t]{\columnwidth}\centering\vspace{0pt}}
  \newcommand{\pairfigmid}{\end{minipage}\hfill
    \begin{minipage}[t]{\columnwidth}\centering\vspace{0pt}}
  \newcommand{\pairfigend}{\end{minipage}\end{figure*}}
\else
  \newcommand{\pairfigstart}{\begin{figure}[!ht]\centering}
  \newcommand{\pairfigmid}{\end{figure}\begin{figure}[!ht]\centering}
  \newcommand{\pairfigend}{\end{figure}}
\fi

\newsavebox{\icfitbox}
\newcommand{\fitbox}[2]{%
  \sbox\icfitbox{#2}%
  \ifdim\wd\icfitbox>#1\relax
    \resizebox{#1}{!}{\usebox\icfitbox}%
  \else
    \usebox\icfitbox
  \fi}

\begin{document}

\ifnum\fullversion=0
\title{\papertitle}
\else
\title{{\bf \papertitle}}
\fi

\ifnum\fullversion=0
\author{\IEEEauthorblockN{Paper \#\,XXX}
\IEEEauthorblockA{Anonymous Submission}}

\IEEEoverridecommandlockouts
\makeatletter\def\@IEEEpubidpullup{6.5\baselineskip}\makeatother
\IEEEpubid{\parbox{\columnwidth}{
		Network and Distributed System Security (NDSS) Symposium 2026\\
		23 - 27 February 2026 , San Diego, CA, USA\\
		ISBN 979-8-9919276-8-0\\
		https://dx.doi.org/10.14722/ndss.2026.[23$|$24]xxxx\\
		www.ndss-symposium.org
}
\hspace{\columnsep}\makebox[\columnwidth]{}}
\else
\author[1]{Paul Gerhart\orcidlink{0000-0002-0164-0187}$^{*}$}
\author[2]{Nadav Kohen\orcidlink{0009-0009-4891-5342}$^{*}$}
\author[3]{Jesse Posner\orcidlink{0009-0006-3020-0769}}
\author[2]{Matias Furszyfer\orcidlink{0009-0003-8946-2040}}
\affil[1]{TU Wien, Vienna, Austria}
\affil[2]{Chaincode Labs, New York, United States}
\affil[3]{Vora, San Francisco, United States}
\date{}
\fi

\maketitle

\ifnum\fullversion=1
\renewcommand{\thefootnote}{\fnsymbol{footnote}}
\footnotetext[1]{Authors contributed equally to this work.}
\renewcommand{\thefootnote}{\arabic{footnote}}
\fi

\begin{abstract}
	The Bitcoin Lightning Network secures hundreds of millions of dollars, yet channel endpoints rely on vulnerable single online keys.
Although threshold signatures are routinely used to protect on-chain Bitcoin, no practical deployment has been possible for Lightning channels.
This is because thresholdizing a Lightning party requires nesting a threshold signature scheme inside of an established two-party MuSig2 protocol without altering its nonce exchange or message flow.

In this work, we resolve this limitation by formalizing nested threshold multi-signatures, a new cryptographic primitive for thresholdizing one participant inside a multi-signature protocol.
As an instance of this primitive, we present Iceberg, the first construction for nested threshold MuSig2 signatures.
Iceberg enables one side of a Lightning channel to operate as a $t$-of-$n$ threshold group while appearing to the counterparty as a standard MuSig2 participant. 
As a result, threshold custody can be deployed unilaterally on today's Lightning Network without requiring any modifications to Bitcoin, the Lightning protocol, or channel counterparties.

We prove the security of Iceberg, integrate a prototype into a production Lightning node, and benchmark its performance. 
Our measurements show that thresholdizing a Lightning channel incurs only modest overhead, since a threshold group tolerating one corrupted member sustains over $93\%$ of the payment throughput of an unmodified endpoint.
\end{abstract}

\ifnum\fullversion=0
\begin{IEEEkeywords}
\paperkeywords
\end{IEEEkeywords}

\IEEEpeerreviewmaketitle
\else
\begin{quote}\small\textbf{Keywords.} \paperkeywords\end{quote}
\fi

\section{Introduction}
\label{sec:intro}

The Bitcoin Lightning Network~\cite{PD16}, a payment-channel layer that settles Bitcoin off-chain, now carries \emph{real money} at \emph{institutional scale}. The value locked in its public channels reached a record of roughly 5,600 BTC, about \$490 million, in December 2025~\cite{lnstats}. This value is held across on the order of 15,000 public nodes and tens of thousands of channels~\cite{lnnodes}, and because most channels are private these figures are \emph{lower bounds}. Major exchanges including Bitfinex, OKX, Kraken, Binance, and Coinbase now settle customer deposits and withdrawals over Lightning~\cite{lnexchanges}, which is cheaper and faster than settling on-chain and increasingly what customers expect.

However, every payment-channel in the Lightning Network rests on a fragile foundation. 
On each side of the channel, there is a single online key holding the channel endpoint's funds, and a single key is a single point of failure. 
Relying on a single key can fail in two ways. The key can be \emph{compromised}: in 2019, attackers drained 7,000 BTC from a single exchange hot wallet in one transaction that passed every security check~\cite{Binance19}. Or the key can be \emph{lost}: hundreds of millions of dollars sit permanently beyond reach on one encrypted drive whose owner forgot the password~\cite{Thomas21}. 
Neither is recoverable, and both kinds of failure recur throughout Bitcoin's history. 
This is already a serious problem for Bitcoin custody. In Lightning, it becomes even more acute, since keeping a channel live and routing payments requires the key to stay online, ruling out the protection of cold storage.

On-chain, a single signing key being a single point of failure is largely a solved problem. Threshold signing~\cite{Desmedt94} ($t$-of-$n$) distributes signing authority across $n$ parties so that any $t$ of them can sign while no single key corruption is fatal, because the compromise of up to $\thres-1$ shares does not allow forging signatures, and similarly the loss of up to $n-\thres$ shares is survivable, as long as $\thres$ shares remain. Threshold custody is standard practice for institutions and serious self-custodians~\cite{CCS:LinNof18,FCW:SwaPoi21}. Bitcoin supports it natively in Script by listing several keys and requiring a quorum of signatures~\cite{BIP11}. Modern Schnorr-based schemes such as MuSig2~\cite{NRS21} (specified in BIP~327~\cite{BIP327}) and FROST~\cite{SAC:KomGol20} achieve the same function more compactly and privately, collapsing the whole group into one aggregate key and signature.
However, the protection that on-chain custody takes \emph{for granted is missing} for Lightning channels. 

To see why, consider how a Lightning channel works (\Cref{fig:channel}). 
%
\begin{figure*}
  \centering
  \resizebox{\linewidth}{!}{%
  \begin{tikzpicture}[
    font=\small,
    lane/.style={rounded corners=4pt},
    badge/.style={circle, draw, thick, inner sep=0pt, minimum size=4.6mm,
                  font=\bfseries\footnotesize, fill=white},
    party/.style={draw, rounded corners=2pt, fill=gray!12, minimum width=1.05cm,
                  minimum height=0.5cm, font=\small},
    musig/.style={draw, thick, rounded corners=2pt, fill=gray!18, draw=black!55,
                  minimum width=1.1cm, minimum height=0.56cm, font=\footnotesize},
    sig/.style={draw, thick, rounded corners=2pt, minimum width=0.88cm,
                minimum height=0.54cm, font=\small},
    sigon/.style={sig, fill=orange!16, draw=orange!65!black},
    sigoff/.style={sig, fill=teal!16, draw=teal!55!black},
    tx/.style={draw, thick, rounded corners=2pt, fill=orange!14, draw=orange!65!black,
               align=center, font=\small, minimum height=0.62cm},
    arr/.style={-{Latex[length=1.6mm]}, semithick},
    bcast/.style={-{Latex[length=1.9mm]}, thick, draw=orange!60!black},
    ss/.style={font=\scriptsize, inner sep=1pt, text=black!60},
    tag/.style={font=\scriptsize\itshape, text=black!55, align=center},
    sep/.style={dashed, line width=0.7pt, draw=black!40},
  ]
    \begin{scope}[on background layer]
      \fill[lane, orange!6] (0,1.15) rectangle (16.3,2.22);
      \fill[lane, teal!6]   (0,-0.98) rectangle (16.3,1.00);
      \draw[sep] (5.40,-0.62) -- (5.40,2.22);
      \draw[sep] (11.00,-0.62) -- (11.00,2.22);
    \end{scope}
    \node[rotate=90, font=\bfseries\footnotesize, text=orange!55!black] at (-0.32,1.68)
       {On chain};
    \node[rotate=90, font=\bfseries\footnotesize, text=teal!55!black, align=center] at (-0.30,0.05)
       {MuSig2};

    \node[badge] at (0.70,2.44) {1};
    \node[anchor=west, font=\bfseries] at (1.05,2.44) {Open};
    \node[badge] at (5.80,2.44) {2};
    \node[anchor=west, font=\bfseries] at (6.15,2.44) {Update};
    \node[badge] at (11.20,2.44) {3};
    \node[anchor=west, font=\bfseries] at (11.55,2.44) {Close};

    \node[tx] (TX1) at (3.55,1.68) {tx: Alice, Bob $\rightarrow \tX$};
    \node[tag, text=orange!45!black, text width=4.8cm] at (8.20,1.68)
       {no transactions in between: \\ the channel is invisible on chain};
    \node[tx] (TX2) at (14.55,1.68) {tx: $\tX \rightarrow$ Alice, Bob};

    \node[party] (A1) at (0.62,0.64) {Alice};
    \node[party] (B1) at (0.62,0.04) {Bob};
    \node[musig] (M1) at (2.30,0.34) {$\keyagg$};
    \node[sigon] (G1) at (3.55,0.34) {$\tX$};
    \draw[arr] (A1.east) -- node[ss,above,pos=0.45,inner sep=2pt]{$X_A$}(M1.west);
    \draw[arr] (B1.east) -- node[ss,below,pos=0.45,inner sep=2pt]{$X_B$}(M1.west);
    \draw[arr] (M1) -- (G1);
    \draw[bcast] (G1.north) -- node[ss,right,text=orange!55!black]{output} (TX1.south);
    \node[tag, text width=3.4cm] at (3.05,-0.28)
       {aggregate key for the channel};

    \node[party] (A2) at (6.12,0.64) {Alice};
    \node[party] (B2) at (6.12,0.04) {Bob};
    \node[musig] (M2) at (7.80,0.34) {MuSig2};
    \node[sigoff] at (9.19,0.48) {};
    \node[sigoff] at (9.12,0.41) {};
    \node[sigoff] (G2) at (9.05,0.34) {$\sigma_i$};
    \node[ss, text=black!55] at (9.95,0.56) {$\times k$};
    \draw[arr] (A2.east) -- node[ss,above,pos=0.45,inner sep=2pt]{$s_A$}(M2.west);
    \draw[arr] (B2.east) -- node[ss,below,pos=0.45,inner sep=2pt]{$s_B$}(M2.west);
    \draw[arr] (M2) -- (G2);
    \node[tag, text width=3.9cm] at (8.70,-0.30)
       {off chain and private, revoking the previous balance};

    \node[party] (A3) at (11.62,0.64) {Alice};
    \node[party] (B3) at (11.62,0.04) {Bob};
    \node[musig] (M3) at (13.30,0.34) {MuSig2};
    \node[sigon] (G3) at (14.55,0.34) {$\sigma_{\mathrm{fin}}$};
    \draw[arr] (A3.east) -- node[ss,above,pos=0.45,inner sep=2pt]{$s_A$}(M3.west);
    \draw[arr] (B3.east) -- node[ss,below,pos=0.45,inner sep=2pt]{$s_B$}(M3.west);
    \draw[arr] (M3) -- (G3);
    \draw[bcast] (G3.north) -- node[ss,right,text=orange!55!black]{broadcast} (TX2.south);

    \node[tag, text=black!70, text width=15cm] at (8.15,-0.80)
       {every MuSig2 signature $\sigma$ is indistinguishable from an ordinary single-key
        BIP\,340 Schnorr signature};

    \draw[arr, line width=0.9pt] (0.15,-1.18) -- (16.30,-1.18);
    \node[font=\scriptsize] at (16.02,-1.02) {time};
  \end{tikzpicture}}
  \caption{\textbf{A Lightning channel has one aggregate key and a sequence of MuSig2
    signatures.} Alice and Bob lock funds under one aggregate MuSig2 key
    $\tX=\keyagg(\{X_A,X_B\})$, which the funding transaction
    $\text{Alice},\text{Bob}\rightarrow\tX$ carries in its output. Every balance update and the
    close is a single MuSig2 2-of-2 signature formed from the partial signatures $s_A$ and $s_B$.
    Only the funding and the closing transaction $\tX\rightarrow\text{Alice},\text{Bob}$ reach the
    chain; everything in between stays private and off-chain.}
  \label{fig:channel}
\end{figure*}
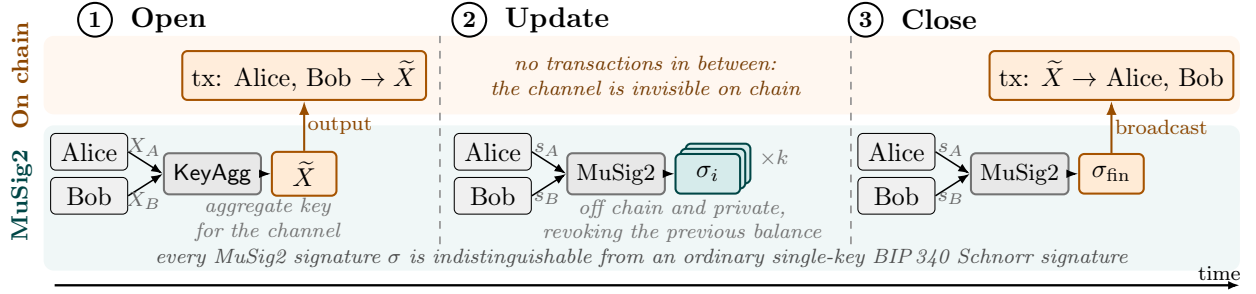
Two parties, say Alice and Bob, open a channel by locking funds into a single output they jointly control, namely a 2-of-2 that, on modern taproot channels~\cite{STC}, takes the form of one aggregate MuSig2 key~\cite{NRS21}. Spending it requires both parties to sign, so neither can move the funds alone. Because it is an aggregate key, the output is indistinguishable on-chain from an ordinary single-key address.
For this reason, the rest of the world cannot even tell a channel is open. From then on, Alice and Bob update their shared balance privately and off-chain, each update re-signing a transaction spending that funding output using MuSig2.
When they decide to close the channel, they broadcast the latest agreed state, spending the funding output and distributing the funds according to the final channel balance, either cooperatively or, if necessary, unilaterally.

Alice's funding key is one half of that 2-of-2 signing key, and on her side, \emph{an online single point of failure she would like to distribute}. 
As of today, Alice \emph{cannot do so}, as the channel's structure stands in the way at multiple stages. 
On the legacy channels Bitcoin first supported~\cite{BOLT3}, the 2-of-2 funding output is a Bitcoin Script, and Alice could in principle widen her half into a threshold script, but not alone, since it would require changing the Lightning protocol both peers run. Furthermore, it would force a non-standard script to be negotiated at every channel open, enlarge every transaction and its fee, and leave an on-chain footprint that advertises her key-management setup, an undesirable privacy leak.
Taproot channels~\cite{BIP341,STC}, which the ecosystem is now adopting for their privacy and smaller transactions, further complicate matters for would-be non-standard script users, since the two sides aggregate into a single MuSig2 key spent by one ordinary signature, and MuSig2 gives each party just one key to contribute, with nothing visible to widen.
Either way, with Lightning moving to taproot channels, a threshold Alice can actually use must live inside one ordinary-looking MuSig2 key. 
In addition, as of today, Alice can only use a threshold if the changes she makes require no changes from Bob, the protocol, or the chain.
Therefore, thresholdizing Alice's side of the channel calls for a scheme that creates threshold partial signatures within an outer MuSig2 session, which is the new primitive we introduce in this work formalized as a \emph{nested threshold multi-signature}. 

\label{sec:intro-barriers}
Finding constructions for a nested threshold multi-signature has been a long-standing open problem for the Bitcoin community~\cite{Ste19,Ste19b,KC23,Sed19,Shl20,Zmn25}.
The main reason why finding a nested threshold multi-signature for MuSig2 is hard is due to the nonce $R$ of a Schnorr signature $(R, s)$ that MuSig2 computes. 
Lightning fixes the session nonce $R$ a full round-trip before the transaction it will sign exists~\cite{STC}, and once Alice has sent her part of it to Bob she cannot revise it. 
This is fine if Alice is a single signer, but if she is thresholdized, her signer set may change in between rounds.
Whichever signing set of Alice's devices is online when the transaction arrives has to complete the signature under a nonce that was fixed before the transaction existed, and possibly chosen by a different signing set. 
In the literature, this issue is solved using stateless threshold signatures such as Arctic~\cite{KG24} using deterministic nonces and replicated secret sharing~\cite{ItoSaiNis87,TCC:CraDamIsh05}. 
Unfortunately, using stateless signatures does not work in the Lightning setting for at least two reasons. 
First, stateless Schnorr signatures (including Arctic) typically derive the nonce from the message to be signed and the current signing session to avoid key leakage that would arise if the two different messages are signed with the same randomness. 
Yet, in the Lightning setting Alice knows neither the message to be signed nor the final signing session when she needs to issue her part of the nonce $R$. 
Second, if the signers in Lightning would be stateless, they could be convinced to sign an outdated channel state which would make them subject to Lightning's punishment mechanisms. 
Hence, any solution for thresholdized Lightning must at least safely carry the latest channel state to be meaningful. 
We discuss the restrictions on nested threshold multi-signatures imposed by the structure of Lightning channels alongside our mitigations of these in~\cref{sec:to}.

Given that threshold custody for Lightning is highly desirable for the Bitcoin community and yet not solved as of today we ask the question:
\begin{quote}
\centering
\emph{Can threshold custody be deployed in today's Lightning Network?
}	
\end{quote}

\subsection{Our Contribution}
\label{sec:intro-contribution}
In this paper, we answer this question in the affirmative. 

\begin{itemize}
	\item We introduce the notion of \emph{nested threshold multi-signatures}, a novel cryptographic primitive designed to $t$-of-$n$ threshold-sign partial signatures of a given multi-signature scheme (\cref{sec:definition}). 
	\item We present the first nested threshold multi-signature \Iceberg, that allows $t$-of-$n$ threshold sign partial MuSig2 signatures~\cite{NRS21,BIP327} (\cref{sec:construction}).
	\item We prove the unforgeability of \Iceberg (\cref{sec:security}) with optimal signing thresholds dictated by the Lightning environment.
		  In particular, unforgeability holds if an adversary compromises up to $t-1$ signing key shares, and any $2t-1$ signing shares are sufficient to produce a valid signature. (Cf.~\cref{sec:to} for a discussion on the optimality of these parameters). 
	\item We show that \Iceberg is a drop-in replacement for MuSig2 at a Lightning endpoint. Its first-round output is exactly the message a single MuSig2 participant would send, and the channel it funds is an ordinary single-key Taproot output.
	In other words, an \Iceberg endpoint is compatible with the Lightning protocol as deployed. As a result, threshold custody can be added unilaterally, with no change to Bitcoin, to the protocol, or to the counterparty.
	\item We integrate \Iceberg into a production Lightning node and measure it over taproot channels~(\cref{sec:evaluation}).
	Making one endpoint a $2$-of-$4$ group adds $3.8$\,ms per payment, which is $6.7\%$ of the work a payment already costs.
	That group only saturates an additional core at $260$ payments per second.
\end{itemize}

\subsection{Related Work}
\label{sec:intro-related}

Threshold signatures~\cite{Desmedt94,STOC:DDFY94,JC:GRJK00,EC:Shoup00a,EC:GJKR96,JC:GJKR07,RSA:GJKR03} have become a practical solution for distributed custody of digital assets, with recent constructions providing efficient signing protocols for Schnorr~\cite{C:Schnorr89}  signatures~\cite{SAC:KomGol20,C:BCKMTZ22,C:CGRS23,C:CriKomMal23,C:CKKTZ25,CiC:Chen25,EC:BDLR25,C:BDLR25,EC:GCRS26,C:KatReiTak24,EC:NioReiTak26,EPRINT:BCLTZ25} and ECDSA~\cite{ACNS:GenGolNar16,CCS:GenGol18,CCS:LinNof18,SP:DKLs19,ESORICS:DOKSS20,EPRINT:GKSS20,CCS:CGGMP20,PKC:CCLST20,SCN:DJNPO20,PKC:YueCuiXie21,SP:ANOSS22}.
However, existing threshold schemes assume that the threshold group constitutes the complete signer.
Lightning creates a fundamentally different setting, in which only one participant of an existing MuSig2 execution is thresholdized, while the remaining participant and the overall signing protocol remain unchanged.
This requires threshold signing to operate inside an existing multi-signature protocol, rather than replacing it.

None of the existing threshold constructions support this setting, since they either bind their nonce generation to the signing set, incur high round complexity, or require the message to be known at nonce creation time (cf.~\Cref{tab:related}).
We restrict this section to the schemes most relevant for deploying threshold custody in Lightning and refer to~\cite{C:CKKTZ25} for a broader survey of Schnorr threshold and multi-signature schemes.

\begin{table}
\centering
\tabfont
\tabcaption{Feature comparison of \Iceberg{} with most relevant threshold-Schnorr and
  nesting schemes (\cmark~yes, \xmark~no).
  $\dagger$~a \cmark{} denotes an assumed honest majority, and MuSig-DN is an $n$-of-$n$ multi-signature, so the axis does not apply~(\nax).}{tab:related}

\setlength{\tabcolsep}{6pt}
\fitbox{\tabscale\columnwidth}{%
\begin{tabular}{@{}lccccc@{}}
\toprule
 & \rot{FROST~\cite{SAC:KomGol20}} & \rot{Arctic~\cite{KG24}} & \rot{MuSig-DN~\cite{CCS:NRSW20}} & \rot{Seurin~\cite{Seu24}} & \rot{\textbf{\Iceberg}} \\
\midrule
Nestable in MuSig2                 & \xmark & \xmark & \xmark & \cmark & \cmark \\
Dynamic online signers             & \xmark & \cmark & \xmark & \cmark & \cmark \\
State-aware                        & \xmark & \xmark & \xmark & \cmark & \cmark \\
Deterministic nonces               & \xmark & \cmark & \cmark & \cmark & \cmark \\
Honest majority$^{\dagger}$        & \xmark & \cmark & \nax   & \cmark & \cmark \\
Complete, proven construction      & \cmark & \cmark & \cmark & \xmark & \cmark \\
Deployable in today's Lightning.   & \xmark & \xmark & \xmark & \xmark & \cmark \\
\bottomrule
\end{tabular}}
\tabcaptionend
\end{table}

The FROST family~\cite{SAC:KomGol20,C:BCKMTZ22,C:CGRS23,C:CKKTZ25} builds its aggregate nonce from the contributions of the quorum that signs, which binds the nonce to that quorum and rules it out for a nested endpoint.
Stateless schemes~\cite{C:GKMN21,C:KonOrlRoy23,PKC:FYZWYX25} remove the binding by deriving the nonce deterministically, and pay for it by taking the message as an input, and Arctic~\cite{KG24} is the most recent of them, and~\cite{CCS:NRSW20,EC:GCRS26} meet the same obstacle. 
As discussed in the introduction, fully stateless signers cannot be used for thresholding Lightning, as signing an old channel state would lead to punishment of the signer set. 
Yet, \Iceberg follows Arctic's pattern of deterministic nonces over replicated publicly-verifiable secret sharing, but seeds the derivation with the state of the channel rather than with the message and the signer set resolving this issue~(\cref{sec:to}).

\Iceberg operates inside a multi-signature protocol, making the multi-signature literature equally relevant to our setting.
Bellare and Neven~\cite{CCS:BelNev06} introduced the first Schnorr multi-signature in the plain public-key model, and MuSig~\cite{DCC:MPSW19} added key aggregation, allowing a group of signers to appear as a single signer on chain.
Achieving two-round Schnorr multi-signatures proved challenging, and Drijvers et al.~\cite{SP:DEFKLN19} showed that natural constructions are vulnerable, while MuSig2~\cite{NRS21}, standardized as BIP~327~\cite{BIP327}, and DWMS~\cite{C:AlpBur21} achieve two rounds by providing each signer with $\nu=2$ nonce commitments. MuSig-DN~\cite{CCS:NRSW20} instead derives nonces deterministically with public verifiability, at the cost of a zero-knowledge proof per signature. Subsequent works further improve the analysis or avoid rewinding-based proofs~\cite{PKC:DOTT21,AC:BelDai21,EC:TesZhu23a,EC:PanWag23,EC:PanWag24,C:BacWag25}.
Lightning uses MuSig2 as its signing protocol and does not permit replacing it with an alternative multi-signature scheme. Therefore, \Iceberg must operate within the MuSig2 execution itself rather than a different underlying construction. 

Placing a threshold on one side of a MuSig2 channel was sketched informally by Seurin~\cite{Seu24} in an unpublished note.
The first formal treatment of nesting a multi-signature scheme into MuSig2 was done by Kohen~\cite{Koh26}.
However Kohen nests a multi-signature $(n$-out-of-$n$) rather than a threshold signature ($t$-out-of-$n$) into MuSig2. 
Since in this setting all signers are required to participate in every signing session, the signer set remains fixed and the availability and compromise challenges that motivate Lightning custody are not addressed. 

\section{Technical Overview} \label{sec:to}
A Lightning channel between two parties locks their coins in a single on-chain output, a \emph{funding output} guarded by a $2$-of-$2$ multisignature. 
The two parties then transact without touching that output again, co-signing a sequence of \emph{commitment transactions} that each encode the current balance, where each new one revokes its predecessor. 
Should a party try to settle the channel on an old, revoked state, Lightning's \emph{penalty mechanism} lets the counterparty claim all of the channel's funds, and that threat keeps both sides committed to the latest balance. 
Bitcoin's Taproot upgrade lets this $2$-of-$2$ be realised with MuSig2~\cite{NRS21}, a two-round multi-signature scheme that aggregates several public keys into a single one. 
The two funding keys aggregate to one key $\tX$, and any spend both parties agree to, in particular a cooperative channel close, is a single Schnorr signature under $\tX$. 
On chain this is a key-path spend indistinguishable from an ordinary single-key payment, since the verifier sees one key and one signature, with no trace of the two parties behind $\tX$. 

Our goal is to turn one side of the channel from a single signer into a group of signers. 
Instead of a single funding key, we want $n$ parties, a user's devices or an exchange's servers, of which any $2t-1$ suffice to sign and which stays secure while fewer than $t$ are compromised. 
Throughout this overview, side $A$ is this group and side $B$ is the channel counterparty. 

Achieving this would be straightforward if we could redesign how Lightning channels work.
Our constraint is that we retrofit it into today's channels, so the scheme can be adopted with no change on the counterparty's side. 
This means $B$ keeps running plain two-party MuSig2 and the chain keeps verifying one Schnorr signature under $\tX$. 
The group must therefore occupy the key slot $A$ holds in the funding $2$-of-$2$ and, at each round, present exactly what a single MuSig2 signer would. 

\Cref{fig:nesting} shows the channel this produces, with side $B$ unchanged and side $A$ a group, and tracks what $B$ and the chain see at each stage. 
At key generation the group's $n$ parties present one group key $X$, which fills side $A$'s slot and aggregates with $X_B$ into the funding key $\tX$ just as two ordinary keys would. 
In the first round the group presents one nonce $R_A$, fixed before the message exists, which aggregates with $R_B$ into the session nonce $R$. 
In the second round the group presents one partial signature $s_A$, which adds to $s_B$ to give the final signature $\sigma = (R, s_A + s_B)$, on chain a single-key spend. 
The group presents one key, one nonce and one partial signature, so at no point can $B$ or the chain tell side $A$ from a lone MuSig2 signer. 
%
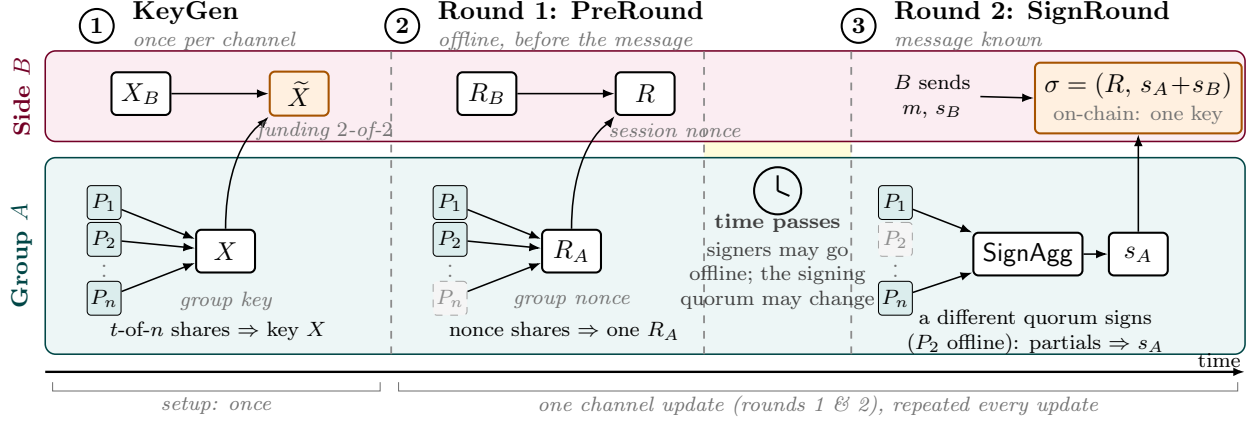
\begin{figure*}[t]
  \centering
  \resizebox{\linewidth}{!}{%
  \begin{tikzpicture}[
    font=\small,
    lane/.style={rounded corners=4pt, line width=0.6pt},
    badge/.style={circle, draw, thick, inner sep=0pt, minimum size=4.6mm,
                  font=\bfseries\footnotesize, fill=white},
    kb/.style={draw, thick, rounded corners=2pt, fill=white,
               minimum width=0.8cm, minimum height=0.58cm, font=\small},
    fund/.style={kb, fill=orange!12, draw=orange!65!black},
    mem/.style={draw, rounded corners=1.5pt, fill=teal!14,
                minimum size=0.46cm, inner sep=1pt, font=\scriptsize},
    memoff/.style={draw, rounded corners=1.5pt, draw=black!35, dashed,
                   fill=black!4, text=black!45, minimum size=0.46cm,
                   inner sep=1pt, font=\scriptsize},
    arr/.style={-{Latex[length=1.6mm]}, semithick},
    tag/.style={font=\scriptsize\itshape, text=black!55, align=center},
    note/.style={font=\scriptsize, align=center},
    sep/.style={dashed, line width=0.7pt, draw=black!45},
  ]
    \begin{scope}[on background layer]
      \fill[yellow!18, rounded corners=2pt] (8.95,-1.64) rectangle (10.95,2.48);
      \draw[lane, draw=purple!60!black, fill=purple!7]  (0,1.25) rectangle (16.3,2.48);
      \draw[lane, draw=teal!55!black,   fill=teal!7]    (0,-1.64) rectangle (16.3,1.03);
      \draw[sep] (4.70,-1.64) -- (4.70,2.48);
      \draw[sep] (8.95,-1.64) -- (8.95,2.48);
      \draw[sep] (10.95,-1.64) -- (10.95,2.48);
    \end{scope}

    \node[rotate=90, font=\bfseries\footnotesize, text=teal!55!black]   at (-0.32,-0.30) {Group $A$};
    \node[rotate=90, font=\bfseries\footnotesize, text=purple!60!black] at (-0.32, 1.86) {Side $B$};

    \node[badge] (h1) at (0.70,2.82) {1};
    \node[anchor=west, align=left] at (1.05,2.82)
       {\textbf{KeyGen}\\[-1.5pt]{\scriptsize\itshape\color{black!55}once per channel}};
    \node[badge] (h2) at (4.85,2.82) {2};
    \node[anchor=west, align=left] at (5.20,2.82)
       {\textbf{Round 1: PreRound}\\[-1.5pt]{\scriptsize\itshape\color{black!55}offline, before the message}};
    \node[badge] (h3) at (11.05,2.82) {3};
    \node[anchor=west, align=left] at (11.40,2.82)
       {\textbf{Round 2: SignRound}\\[-1.5pt]{\scriptsize\itshape\color{black!55}message known}};

    \draw[arr, line width=0.9pt] (0,-1.88) -- (16.3,-1.88);
    \node[font=\scriptsize] at (15.95,-1.72) {time};
    \draw[thin, black!55] (0.10,-2.02) -- (0.10,-2.14) -- (4.55,-2.14) -- (4.55,-2.02);
    \node[tag] at (2.32,-2.34) {setup: once};
    \draw[thin, black!55] (4.80,-2.02) -- (4.80,-2.14) -- (16.20,-2.14) -- (16.20,-2.02);
    \node[tag] at (10.50,-2.34) {one channel update (rounds 1 \& 2), repeated every update};

    \node[mem] (p1) at (0.80, 0.42) {$P_1$};
    \node[mem] (p2) at (0.80,-0.08) {$P_2$};
    \node[font=\scriptsize, text=black!55] at (0.80,-0.47) {$\vdots$};
    \node[mem] (pn) at (0.80,-0.88) {$P_n$};
    \node[kb] (Xk) at (2.45,-0.24) {$X$};
    \draw[arr] (p1) -- (Xk);
    \draw[arr] (p2) -- (Xk);
    \draw[arr] (pn) -- (Xk);
    \node[tag] at (2.45,-0.92) {group key};
    \node[note] at (2.35,-1.32) {$t$-of-$n$ shares $\Rightarrow$ key $X$};
    \node[kb] (XB) at (1.30,1.90) {$X_B$};
    \node[fund] (TX) at (3.45,1.90) {$\tX$};
    \draw[arr] (Xk.north) .. controls (2.45,1.25) and (3.0,1.55) .. (TX.south west);
    \draw[arr] (XB) -- (TX);
    \node[tag] at (3.8,1.40) {funding $2$-of-$2$};

    \node[mem] (q1) at (5.50, 0.42) {$P_1$};
    \node[mem] (q2) at (5.50,-0.08) {$P_2$};
    \node[font=\scriptsize, text=black!55] at (5.50,-0.47) {$\vdots$};
    \node[memoff] (qn) at (5.50,-0.88) {$P_n$};
    \node[kb] (RA) at (7.15,-0.24) {$R_A$};
    \draw[arr] (q1) -- (RA);
    \draw[arr] (q2) -- (RA);
    \draw[arr] (qn) -- (RA);
    \node[tag] at (7.15,-0.92) {group nonce};
    \node[note] at (7.05,-1.32) {nonce shares $\Rightarrow$ one $R_A$};
    \node[kb] (RB) at (6.00,1.90) {$R_B$};
    \node[kb] (Rc) at (8.15,1.90) {$R$};
    \draw[arr] (RA.north) .. controls (7.15,1.25) and (7.70,1.55) .. (Rc.south west);
    \draw[arr] (RB) -- (Rc);
    \node[tag] at (8.55,1.40) {session nonce};

    \draw[thick] (9.95,0.58) circle (3.2mm);
    \draw[thick] (9.95,0.58) -- (9.95,0.79);
    \draw[thick] (9.95,0.58) -- (10.11,0.51);
    \node[note, text=black!75] at (9.95,-0.38)
       {\textbf{time passes}\\[1pt] signers may go\\ offline; the signing\\ quorum may change};

    \node[note] (bn) at (12.05,1.86) {$B$ sends\\ $m$, $s_B$};
    \node[mem]    (r1) at (11.55, 0.42) {$P_1$};
    \node[memoff] (r2) at (11.55,-0.08) {$P_2$};
    \node[font=\scriptsize, text=black!55] at (11.55,-0.47) {$\vdots$};
    \node[mem]    (rn) at (11.55,-0.88) {$P_n$};
    \node[kb, minimum width=1.15cm] (SA) at (13.35,-0.28) {$\signagg$};
    \draw[arr] (r1) -- (SA);
    \draw[arr] (rn) -- (SA);
    \node[kb] (sA) at (14.85,-0.28) {$s_A$};
    \draw[arr] (SA) -- (sA);
    \node[note] at (13.45,-1.34) {a different quorum signs\\ ($P_2$ offline): partials $\Rightarrow$ $s_A$};
    \node[fund, align=center, minimum width=2.45cm] (SIG) at (14.85,1.84)
       {$\sigma=(R,\,s_A\!+\!s_B)$\\[-1pt]{\scriptsize\normalfont\color{black!55}on-chain: one key}};
    \draw[arr] (sA.north) -- (SIG.south);
    \draw[arr] (bn.east) -- (SIG.west);
  \end{tikzpicture}}
  \caption{A thresholdized Lightning channel. \textbf{(1)~KeyGen} (once per channel)
    fixes the group key $X$ from replicated VPSS shares of $P_1,\dots,P_n$; $X$ fills one
    slot of the funding $2$-of-$2$ $\tX=\keyagg(\{X,X_B\})$, on chain a single key-spend.
    Each update then runs two rounds: \textbf{(2)~PreRound}, offline and message-free,
    turns the session identifier $\sid$ into one group nonce $R_A$ (the same for every
    quorum) that joins $B$'s $R_B$ into $R$; \textbf{(3)~SignRound}, once $m$ is known, has
    the online quorum output partials that $\signagg$ combines into $s_A$, yielding the
    MuSig2 signature $\sigma=(R,s_A+s_B)$ under $\tX$. Between the rounds signers may go
    offline, so the quorum that precommits the nonce need not be the one that signs.}
  \label{fig:nesting}
\end{figure*}

\paragraph{Our setting}
The Lightning protocol and the practical deployment context together shape what such a group can look like. 
Some of the following constraints may seem cherry-picked, but the rest of this section shows that they are already required of a threshold lightning construction. 

\begin{constraintlist}
\item\label{cstr:offline} \textbf{Signers go offline.} 
Keys may sit in cold storage, so an honest member is routinely absent rather than malicious. 
The quorum that begins a signing session need not be the one that finishes it, and absent members cannot simply be counted as corruptions. 

\item\label{cstr:premsg} \textbf{Nonces precede the message.} 
In the simple taproot channels protocol~\cite{STC}, each side exchanges its MuSig2 nonce a full round-trip before the commitment transaction it will sign even exists. 
The group must therefore commit to a single nonce $R_A$ before knowing the message $m$, and since the signing quorum may not yet be fixed at that point, the nonce cannot depend on who will eventually sign. 

\item\label{cstr:nonce-fixed} \textbf{Nonces cannot change.} 
Once $R_A$ is sent to the counterparty it is fixed for this session. The outer MuSig2 session is unmodified BIP~327, so the group cannot supply a different nonce after the fact. 
Should the quorum that finishes the second round differ from the one that started the first, the nonce already committed to $B$ must stay the same. 

\item\label{cstr:consensus} \textbf{The group must agree on the live state.} 
Were a corrupt minority of $t-1$ members to convince an honest one to sign a superseded commitment transaction, the counterparty would invoke the penalty mechanism and claim all channel funds. 
The group must therefore reach consensus on the current commitment number before signing. 

\item\label{cstr:small} \textbf{Signer sets are small.} 
A typical group is a user's few devices or the handful of keys an exchange distributes across staff and servers, and many more would be operationally unwieldy. 
We therefore expect $n$ to be at most a dozen.
\end{constraintlist}

Offline signers and small signer sets are practical assumptions rather than protocol constraints, but they are realistic for the deployments we envision. 
Consensus on the live state carries an immediate consequence for the group's parameters. Agreeing with up to $t-1$ faulty members is Byzantine agreement, which requires a strict honest supermajority and thus $n \geq 3t-2$.
\Cref{tab:thresholds} lists, for each \Iceberg threshold up to a group of twelve, the corruptions it tolerates, the online quorum it needs, and the smallest group that supports it.

\begin{table}
\centering
\tabfont
\tabcaption{Possible \Iceberg thresholds. 
A threshold $t$ tolerates up to $t-1$ corruptions, needs a signing quorum of $2t-1$ online members, and requires a group of at least $n \geq 3t-2$.}{tab:thresholds}

\setlength{\tabcolsep}{8pt}
\fitbox{\columnwidth}{%
\begin{tabular}{@{}cccc@{}}
\toprule
Threshold $t$ & Corruptions & Signing quorum & Minimal group \\
\midrule
2 & 1 & 3 & 4 \\
3 & 2 & 5 & 7 \\
4 & 3 & 7 & 10 \\
5 & 4 & 9 & 13 \\
\bottomrule
\end{tabular}}
\tabcaptionend
\end{table}

\paragraph{Strawman approaches} Now that we have established the setting, let us explore why existing threshold and multi-signature schemes do not present solutions to the task of constructing a thresholdized Lightning signer. 
We consider two strawman approaches, MuSig2~\cite{NRS21} and FROST~\cite{SAC:KomGol20}. 

MuSig2 seems like the natural starting point, since the channel already uses it. 
And indeed one MuSig2 instance can be nested inside another, since the inner group's aggregate key is itself an ordinary public key and can occupy side $A$'s slot while the group runs MuSig2 internally~\cite{Koh26}. 
This does not solve our problem, though, because nesting only splits side $A$ into an $n$-of-$n$ group, where every one of the $n$ parties must be online for each signing. 
Tolerating offline signers is one of our goals as is improving key management, so a multi-signature is not enough, and we need a genuine $t$-of-$n$ threshold. 

This calls for a threshold signature, which produces a signature from any $t$ of the $n$ parties and lets the absent ones stay offline. 
We try FROST~\cite{SAC:KomGol20}, the standard such scheme, to see where threshold Schnorr signatures break when nested into Lightning. 
To sign in FROST, each signer commits two nonces $(D_i, E_i)$. A signing quorum $C$ forms the aggregate nonce $R = D\,E^{b}$ from the products $D, E$ of these commitments and a binding factor $b$ that hashes the quorum $C$ and the message $m$, and the partial signatures are combined by Lagrange interpolation over $C$. 
The nonce $R$ therefore depends on which members form the quorum $C$, and this is not specific to FROST. 
Most threshold Schnorr schemes relying on Shamir secret sharing~\cite{Shamir79} build the aggregate nonce from fresh per-signer contributions $R_i$ whose discrete logarithm only the contributing signer knows, so the nonce is tied to the particular set of signers that produced it.
A different quorum brings different contributions and reconstructs a different $R$~\cite{C:CriKomMal23,SAC:KomGol20}. 
The group commits the single nonce $R_A$ to the counterparty in the first round, but in Lightning (and more generally in most threshold-nested signing settings) a different quorum may finish the second round, and it would reconstruct a different $R_A$ than the one already promised to $B$. 

\paragraph{Design decision 1, replicated secret sharing} 
FROST and its relatives~\cite{SAC:KomGol20} violate two requirements that nesting places on the group's nonce. 
First, the nonce must be the same whichever quorum produces it, or the second round would contradict the $R_A$ already committed in the first~(\cref{cstr:nonce-fixed}).
Second, it must remain producible by a member that was offline in the first round, since the finishing quorum need not be the starting one~(\cref{cstr:offline}). With nonces committed or shared in the first round, an absent member holds nothing and cannot sign in the second.
We meet both by deriving the nonce \emph{deterministically} from \emph{replicated} shares. 

Replicated secret sharing~\cite{ItoSaiNis87} splits the secret into one share $\phi_{\aset}$ for each subset $\aset$ of $t-1$ members, distributed to every member \emph{outside} $\aset$, so that the secret is their sum $\sk = \ssum_{\aset} \phi_{\aset}$.
Any subset of $t-1$ members, $\aset$, miss precisely the share indexed by $\aset$ and learn nothing about the aggregate secret, while any larger set holds every summand $\phi_{\aset}$.
Because the secret is this one fixed sum and every qualifying quorum holds all of its terms, reconstruction does not depend on who is present, unlike the per-signer aggregates of FROST or Sparkle that are assembled from whoever happens to sign. 
 
Instead of sampling a fresh nonce and sharing it in the first round, each member derives its contribution from its long-lived shares through a fixed pseudorandom function, evaluated on a per-session input, in the manner of pseudorandom secret sharing~\cite{TCC:CraDamIsh05}.
Nothing is exchanged in the first round, so a member absent then can still recompute its contribution in the second from shares it has held since key generation~\cite{TCC:CraDamIsh05}. 

Put together, every quorum feeds the same shares into the same function and (given the same seed) obtains the same nonce, available to any qualifying set at any time. 

The primary trade-off of this approach is storage complexity.  
Each member holds one share per $(t-1)$-subset it lies outside of, resulting in $\binom{n-1}{t-1}$ shares.
As the number of shares grows quickly with $n$, replicated sharing is  dismissed as impractical for large deployments. 
In our setting it is not, because we expect signer sets to be small~(\cref{cstr:small}), so the share count stays modest. 

\paragraph{Design decision 2, unique pre-message nonces} 
A deterministic nonce is a function of its seed, so every nonce requirement is a seed requirement. 
A deterministic nonce is safe only if it is never reused across two different messages. If a group were to sign two messages under one nonce $R$, the two partial signatures $s = r + c\,x$ and $s' = r + c'\,x$ on challenges $c \neq c'$ would, by the two-special-soundness of Schnorr, expose the signing key as $x = (s-s')/(c-c')$. 
Therefore, the seed must vary whenever the challenge does. 

The literature derives the nonce from the message itself~\cite{KG24}, so an identical message gives an identical nonce and signing the same message twice merely repeats one signature rather than exposing a key. 
We cannot follow this, because the nonce is fixed in the first round, before the message, the commitment transaction to be signed, exists~(\cref{cstr:premsg}).
Instead we use our Lightning protocol setting's state as part of our seed. Every commitment transaction carries a commitment number that is unique and strictly increasing over the channel's life, and it is known before the transaction is assembled. 
We seed the nonce with this number as a part of the session identifier $\sid$, giving a nonce that is fixed before the message yet still varies from one channel state to the next, at no added cost. 

Since the derivation depends on $\sid$ alone, two commitment transactions that carry the same commitment number share a nonce, and the key leakage above returns.

\paragraph{Design decision 3, consensus on the live state} 
The gap left by deterministic nonces is that one $\sid$ may be signed under two messages, as nothing in the derivation prevents the group from being driven to sign two different commitment transactions under the same commitment number.
If that happens, the two partial signatures share the nonce $R$ but carry distinct challenges, and that leaks the signing key. 

Nesting makes this worse, because the outer session runs MuSig2 with $\nu = 2$ nonces per signer.
Two signatures on one $\sid$ do not merely repeat but place the attacker in the concurrent ROS setting~\cite{EC:BLLOR21} that $\nu = 2$ exists to resist, with the group as the target.
The defence against that setting is more nonces, and we cannot add them, since the outer session is unmodified BIP~327.
We therefore cannot prevent the reuse within our construction.

Instead, we use a mechanism any thresholdized channel already relies on.
The group must agree on the current commitment number regardless of our scheme~(\cref{cstr:consensus}), since signing a superseded state lets the counterparty invoke the penalty mechanism and seize the channel.
An honest member therefore signs at most one message under any one $\sid$.
The bound $n \geq 3t-2$ secures an honest supermajority, which keeps a corrupt minority from forcing a second one, preventing nonce reuse.
If an implementation does want to retry a failed signing attempt at the top-level MuSig2, it has to agree on the protocol messages that change the $\sid$.

The channel already requires this agreement, so this method of keeping our deterministic nonces safe comes with no additional cost for our construction.

\paragraph{\Iceberg{}} 
Putting these three design choices together yields a $t$-of-$n$ signing protocol that interacts with the outer session like a single MuSig2 participant. 
At key generation, the group's secret key $\sk$ is shared among the members with replicated secret sharing~\cite{ItoSaiNis87}. 
The key is split additively into $\binom{n}{t-1}$ summands $\sk = \ssum_{\aset} \phi_{\aset}$, one field element $\phi_{\aset} \in \Zpp$ per $(t-1)$-sized subset $\aset$ of the members, and each $\phi_{\aset}$ is handed to every member \emph{outside} $\aset$.
A member thus holds the $\binom{n-1}{t-1}$ summands whose subset excludes it, and any $t$ members jointly hold all summands and so could reconstruct $\sk$, while any $t-1$ miss the summand for their own subset and learn nothing.
The members' shares determine a single group key $X = g^{\sk}$, which occupies side $A$'s slot in the funding $2$-of-$2$. 
The shares can be set up by a trusted dealer or a distributed key generation, and the same shares later serve both keys and nonces. 

The first round runs offline for the current session identifier $\sid$, before the message exists.
For each of the two nonce slots $i \in \{1,2\}$, every member evaluates a pseudorandom function on $i \concat \sid$ under each of its shares, $\Hprf(\phi_{\aset}, i \concat \sid)$, and combines the results into its two nonce shares $(R_{1,k}, R_{2,k})$, which it broadcasts.
A present quorum first checks the shares for well-formedness, that the replicated copies of each $\phi_{\aset}$ produced the same value, and then aggregates them into the group's two nonces, bound together into the nonce $R_A$ handed to the counterparty.
Since the nonce shares are a deterministic function of $\sid$ alone and every share $\phi_{\aset}$ is replicated across all members outside $\aset$, any qualifying quorum recomputes the same $R_A$, so a member absent in this round is not locked out of the next.

In the second round, once the message $m$ and the counterparty's nonce arrive, each member forms the binding factor $\check{b}$ and the challenge $c$, both computable from the public transcript, and returns its partial signature 
\[
  s_k = r_{1,k} + r_{2,k}\,\check{b} + c\,a\,x_k, 
\]
where $r_{i,k}$ are its nonce-share scalars from the first round and $x_k$ its signing-key share. 
Lagrange interpolation over the quorum turns the $\{s_k\}$ into side $A$'s single partial signature $s_A$, which the outer MuSig2 session adds to $s_B$ to form the final signature. 
The protocol is two rounds, and the first is precomputable before the message exists. This is the same as the round structure of MuSig2 itself, so the group adds no rounds to the channel protocol. 

\subsection{Proving Security}
Proving security for a nested threshold multi-signature cannot be done with existing security models, as nested threshold multi-signatures neither fit pure threshold nor standard multi-signature frameworks. 
Instead, they represent a threshold group occupying a single slot within an outer multi-signature scheme.  
Threshold unforgeability games verify a forgery under the group's own key, so the group is the entire signing side and no outer session exists. 
The recent work of Kohen nests a MuSig2 session inside of MuSig2~\cite{Koh26}, where the inner signing session verifies against an outer aggregate key.
We need exactly that, but the nested signer there is itself a multi-signature over a fixed set of members, and every one of them signs every session. 
Nesting a threshold signature adds three things that \cite{Koh26} does not model, namely a corruption threshold inside the nested signer, a quorum that may differ between the two rounds~(\cref{cstr:offline}), and agreement on session state that the channel forces on the group~(\cref{cstr:consensus}). 
We define unforgeability for nested threshold multi-signatures by modifying the security model of~\cite{Koh26} accordingly, so that the adversary controls every outer co-signer, corrupts up to $t-1$ members of the group, and opens concurrent signing sessions on messages and outer contexts of its choosing.
In addition, the adversary can decide to switch the signing quorum from one round to the next. 
The adversary breaks unforgeability of a nested threshold multi-signature if it can produce a signature that verifies under an outer aggregate key without having collected enough honest contributions from the remaining honest parties to have assembled one itself.

Our final goal is now to prove security of \Iceberg with this new security model. 
First, we observe that by using replicated secret sharing to make each protocol round quorum-independent, the $t$-out-of-$n$ shared messages exchanged during an \Iceberg{} round can be mapped directly to an \emph{additive} combination of shares. 
Recall that under replicated secret sharing, a secret key summand is defined for each $(t-1)$-sized subset of signers, and each party receives precisely the shares corresponding to the subsets it does \emph{not} belong to~\cite{ItoSaiNis87}. 
This yields $\binom{n}{t-1}$ underlying key summands in total. 
As all protocol messages in \Iceberg{} are deterministicly computed from the shared key, the messages exchanged among any valid quorum in \Iceberg{} can be deterministically expanded into $\binom{n}{t-1}$ individual protocol messages. 
This mapping allows us to convert the $t$-out-of-$n$ threshold signature \Iceberg{} into a virtual $\binom{n}{t-1}$-out-of-$\binom{n}{t-1}$ multi-signature scheme, which we denote as \scheme{0}. 
In~\cref{lem:sharconv}, we formalize this share conversion and prove that any adversary breaking the unforgeability of \Iceberg{} can be transformed into an efficient adversary breaking the unforgeability of \scheme{0}. 

It is thus sufficient to show that \scheme{0} is unforgeable. 
We achieve this via a reduction from \scheme{0} to the nested multi-signature scheme $\nestedmusig$~\cite{Koh26}. 
Specifically, we embed the $\eufcma$ challenge key of the honest signer in the security experiment of $\nestedmusig$ into a single secret summand of \scheme{0} unknown to the threshold adversary. 
We then simulate all remaining computations by combining honest evaluations over known summands with queries to the $\eufcma$ signing oracle for the challenge summand. 
Careful programming of the aggregation hash function $\Hagg$ ensures that the partial signatures returned by the oracle match the expected distributions in \scheme{0}. 
Finally, because $\nestedmusig$ comes with security proofs in both the Random Oracle Model (ROM, using $\nu=8$ at depth 2) and the ROM + Algebraic Group Model (AGM, using $\nu=2$ as in BIP 327), \scheme{0} (and by extension \Iceberg{}) inherits both security bounds. 

In all, we show that \Iceberg{} fits threshold Lightning's operational constraints demonstrating that provably secure threshold custody for Lightning channels is achievable.

\section{Preliminaries} \label{sec:prelims}

\mypar{Notation}
We write $x \gets y$ for assignment and $x \sample X$ for uniform sampling.
A group description $\gparam$ denotes a cyclic group $\GG$ of prime order $p$ with generator $g$.
For $n \in \NN$ we write $\interval{n} = \{1,\ldots,n\}$, $\concat$ for the concatenation of bitstrings, and $\verts{S}$ for the cardinality of a set $S$.
For a group of $n$ members with threshold $t$, subsets of size $t-1$ are set in bold, so $\asets = \sbinom{\interval{n}}{t-1}$ is the family of all $(t-1)$-subsets of $\interval{n}$, and $\asetsof{k} = \{\aset\in\asets : k\notin\aset\}$ are the subsets whose shares member $k$ holds.

\paragraph{Communication model}
We work in the synchronous communication model. 
The endpoints of the higher-level multi-signature are connected by pairwise authenticated channels. 
The parties of the thresholdized endpoint are likewise connected by pairwise authenticated channels. 
They further use an untrusted aggregator node, which may be any of the participants or another party that facilitates communication (though this leaks privacy if another party is used).

\mypar{Multi-signatures}
A multi-signature scheme lets $n$ signers, each holding an ordinary key pair, jointly produce a single signature verifiable under one aggregate key.
It provides a key generation $\keygen$ that outputs a signer's key pair, a key aggregation $\keyagg$ that maps a list of public keys to one aggregate key $\tX$, an interactive signing protocol, and a verification $\ver$.
The signature is an ordinary one under $\tX$, so on chain the group is indistinguishable from a single key.

\mypar{MuSig2}
MuSig2~\cite{NRS21} instantiates this over $\GG$ and is the scheme Lightning uses.
Signer $i$ holds $X_i = g^{x_i}$, and $\keyagg$ maps a key list $L$ to $\tX = \sprod_i X_i^{a_i}$ with coefficients $a_i = \Hagg(L, X_i)$ that bind each key to the whole list.
Signing runs in two rounds.
The first round is message-independent, and each signer samples $\nu$ nonces and sends $R_{i,j} = g^{r_{i,j}}$ for $j \in \interval{\nu}$ (BIP~327 uses $\nu = 2$), and the nonces are aggregated slot-wise to $R_j = \sprod_i R_{i,j}$.
In the second round, once the message $m$ is known, a binding factor $b = \Hnon(\tX, (R_1,\ldots,R_\nu), m)$ combines the aggregates into the session nonce $R = \sprod_j R_j^{b^{j-1}}$, and each signer returns the partial signature $s_i = c\,a_i\,x_i + \ssum_j r_{i,j}\,b^{j-1}$ for the challenge $c = \Hsig(\tX, R, m)$. The partials sum to $s$, and $\sigma = (R, s)$ verifies as a Schnorr signature under $\tX$.
Since $\tX$ is itself an ordinary public key, one MuSig2 aggregate key can occupy a signer slot of another, a \emph{nested} MuSig2~\cite{Koh26}, which is the structure \Iceberg uses to hide a group inside one side of the channel.
We recall the pseudocode of the nested scheme $\nestedmusig$ in \cref{fig:nestedmusig} in \cref{sec:appdx:prelims}.

\mypar{Unforgeability}
A (nested) multi-signature is EUF-CMA secure if no efficient adversary that controls every signer but one honest signer, and that opens concurrent signing sessions of its choice, can forge a signature under an aggregate key involving the honest key on a message it never asked the honest signer to sign.
We recall the formal game in \cref{fig:game:eufcma} in \cref{sec:appdx:prelims}.
Our security reduces to the EUF-CMA security of $\nestedmusig$, which holds under standard assumptions~\cite{Koh26}, and \cref{sec:security} makes the reduction and its instantiation precise.

\mypar{Verifiable pseudorandom secret sharing}
A verifiable pseudorandom secret sharing (VPSS) \cite{KG24} is a $t$-of-$n$ secret sharing whose long-lived shares let the members repeatedly derive sharings of pseudorandom values, one for each public tag $w$, without interaction.
It provides a key generation $\keygen(n, t, \mu)$ that distributes long-lived shares $\sk_1,\ldots,\sk_n$ of a secret $\sk$, a share derivation $\gen(k, \sk_k, w)$ that outputs member $k$'s share for tag $w$, and algorithms $\vpssver{1}$, $\agg$, and $\recover$ that check the shares of a quorum and interpolate the derived value in the exponent or in the scalar, where the quorum parameter $\mu$ fixes how many shares verification needs.
We require three properties. Under \emph{verifiability} a quorum can check that the shares it received are consistent, under \emph{uniqueness} every qualifying quorum obtains one and the same value for each tag, and under \emph{pseudorandomness} the derived values are indistinguishable from random to any coalition of at most $t-1$ members.

\mypar{A VPSS from replicated secret sharing}
We use the construction of~\cite{KG24}, in which every share doubles as a key of a pseudorandom function.
Key generation splits the secret $\sk = \ssum_{\aset\in\asets}\phi_{\aset}$ into one summand $\phi_{\aset}\sample\Zpp$ per subset $\aset\in\asets$ and hands $(\aset, \phi_{\aset})$ to every member outside $\aset$, so member $k$ holds $\sk_k = \{(\aset,\phi_{\aset})\}_{\aset\in\asetsof{k}}$.
Given a tag $w$, $\gen$ evaluates $\Hprf$ under each held summand and combines the results into member $k$'s share. Across members these are Shamir shares of degree $t-1$ of the value $\ssum_{\aset}\Hprf(\phi_{\aset}, w)$.
Two features of this construction drive our design.
First, $\gen$ is deterministic, so a member can derive its share for any tag at any time from $\sk_k$ alone.
Second, for a quorum parameter $\mu \geq 2t-1$, verifiability and uniqueness hold information-theoretically, since a quorum of $\mu$ members contains at least $t$ honest shares, which determine the degree-$(t-1)$ polynomial and expose every inconsistent share.
Whenever $\vpssver{1}$ accepts, $\agg$ therefore returns $g^{\ssum_{\aset}\Hprf(\phi_{\aset}, w)}$, the same value for every qualifying quorum.
We recall the scheme, denoted $\vpss{1}$, in \cref{fig:vpss1} in \cref{sec:appdx:prelims}, and a variant $\vpss{0}$ without share conversion appears in \cref{sec:security}.

\section{Nested Threshold Multi-Signatures}
\label{sec:definition}

We propose a new cryptographic primitive, namely \emph{nested threshold multi-signatures}. 
Unlike conventional threshold signatures, which replace an entire signer with a threshold group, a nested threshold multi-signature allows a $t$-of-$n$ group of members to occupy a single signer slot of an outer multi-signature protocol.
This enables threshold signing within an existing multi-signature execution while leaving the outer protocol unchanged.

We define the primitive with respect to a two-round multi-signature scheme~(\cref{sec:prelims}) with key aggregation $\keyagg$ and verification $\ver$, whose signing session we call the \emph{outer session}.
We state syntax and security for one level of nesting, the setting of \cref{sec:to}, and deeper nestings generalize as in~\cite{Koh26}.

\paragraph{Syntax}
A nested threshold multi-signature scheme extends a two-round threshold signature scheme with interfaces that connect the inner threshold execution to the outer multi-signature session.
We follow the two-round threshold-signature formalization of~\cite{C:CKKTZ25} and the multi-signature formalization of~\cite{CCS:NRSW20}. 
Nesting changes the interfaces in three places.
First, each round ends in an aggregation algorithm that condenses the members' shares into the single message the outer session expects from the group's slot.
Second, the signing round receives the entire outer context in which the group signs.
Third and most importantly, the two rounds do not share per-member state, and they are linked only through a session identifier $\sid$, so a member absent in the first round can still serve in the second~(\cref{cstr:offline}). 
\begin{figure}[!t]
 \centering
 \shrinkfit{\gamescale}{%
 \begin{varwidth}{2\columnwidth}
 \begin{pcvstack}[boxed,center]
  \pcsetargs{linenumbering}
  \procedure{Game $\pcgame{\tsmseufcma}{\Sigma}$}{%
   \rule{0pt}{\baselineskip}
   \param \gets \setup(\secparam) \\
   (n, t, \corrupt) \gets \adv(\param)\\
   \pcassert \verts{\corrupt} < t\\
   \honest \defeq \interval{n}\setminus\corrupt\\
   PK \defeq \emptyset\pcsc SK\defeq \emptyset\pcsc Started \defeq \emptyset\pcsc SID \defeq \emptyset\pcsc M \defeq \emptyset\pcsc Q \defeq \emptyset\\
   (L^*, m^*,\sigma^*) \gets \adv^{\keygeno,\preroundo,\signo}() \\
   \pcassert \exists i^*, PK[i^*]\in L^*\\
   \tX^* \defeq \keyagg(L^*)\\
   \pcreturn \ver(\tX^*,m^*,\sigma^*)=\tr\land \forall \sid,\verts{Q[m^*, \sid]} < t - \verts{\corrupt}
  }
  \pcvspace
  \procedure{Oracle $\keygeno(i)$}{
    \rule{0pt}{\baselineskip}
    \pcassert i\in\honest\land i\notin Started\\
    Started \defeq Started \cup \{i\}\\
    \pclinecomment{Interactively run $\keygen(i)$ with $\adv$ controlling all network connections}\\
    (\pk_i, \sk_i) \gets \langle \keygen(i),\adv\rangle\\
    PK[i] \defeq \pk_i\pcsc SK[i] \defeq \sk_i\\
    \pcreturn \pk_i
  }
  \pcvspace
  \procedure{Oracle $\preroundo(\sid, i)$}{%
    \rule{0pt}{\baselineskip}
    \pcassert i\in\honest \land PK[i]\neq \bot \land \sid \in \sidspace\\
    \pcassert SID[i][\sid] = \bot \pcsc SID[i][\sid] \defeq \text{ready}\\
    \msg_i \gets \preround(i, SK[i], \sid)\\
    \pcreturn \msg_i
  }
  \pcvspace
  \procedure{Oracle $\signo\left(\sid, i, m, \msg, \ctx\right)$}{%
    \rule{0pt}{\baselineskip}
    \pcassert i\in\honest\\
    \pcassert SID[i][\sid] = \text{ready}\\
    \pcassert M[\sid] \in \{\bot, m\}\pcsc M[\sid] \defeq m \pccomment{consensus on the live state}\\
    s_i \gets \signround\left(i, SK[i], PK[i], \sid, m, \msg, \ctx\right) \\
    SID[i][\sid] \defeq \text{done}\\
    Q[m, \sid] \defeq Q[m, \sid]\cup \{i\}\\
    \pcreturn s_i
  }
 \end{pcvstack}
 \end{varwidth}}
 \caption{The TS-MS-EUF-CMA security game for a nested threshold multi-signature scheme $\Sigma$, where $\keyagg$ and $\ver$ are the key aggregation and verification of the outer scheme. The assertion on $M[\sid]$ models the consensus on the live state~(\cref{cstr:consensus}), so across all honest members, each session identifier is signed on at most one message.}
 \label{fig:game:tsms}
\end{figure}
A nested threshold multi-signature scheme thus consists of the six algorithms $(\setup, \allowbreak \keygen, \allowbreak \preround, \allowbreak \preagg, \allowbreak \signround, \allowbreak \signagg)$ with the following interfaces.
\begin{itemize}
  \item \defitem{$\setup(\secparam) \to \param$:} reads the security parameter and outputs public parameters $\param$, which fix the space $\sidspace$ of admissible session identifiers and are an implicit input to all other algorithms.
  \item \defitem{$\keygen(k, n, t, A) \to (\pk, \sk_k)$:} run by member $k$ on the group size $n$ and the threshold $t$, interacting with the other members through the interface $A$. It outputs the group's public key $\pk$, which fills the group's slot of the outer key list, and $k$'s secret share $\sk_k$.
  \item \defitem{$\preround(k, \sk_k, \sid) \to \msg_k$:} the message-independent first round, which reads a session identifier $\sid \in \sidspace$ and outputs member $k$'s share $\msg_k$ of the group's first-round message.
  \item \defitem{$\preagg(\pk, C, \{\msg_j\}_{j\in C}) \to \msg$:} verifies and aggregates the first-round shares of a quorum $C$ into the group's first-round message $\msg$ in the outer session.
  \item \defitem{$\signround(k, \sk_k, \pk, \sid, m, \msg, \ctx) \to s_k$:} the second round, which reads the message $m$, the group's aggregate first-round message $\msg$, and the outer context $\ctx$, consisting of the group's position in the outer key list, the co-signers' public keys, and the outer session's aggregate first-round message. It outputs member $k$'s partial signature $s_k$.
  \item \defitem{$\signagg(C, \{s_k\}_{k\in C}) \to s$:} combines the partial signatures of a quorum $C$ into the group's second-round message $s$.
\end{itemize}
A nested threshold multi-signature scheme has no verification algorithm of its own. Once the outer session combines the group's second-round message with those of the co-signers, the result is an ordinary signature of the outer scheme under the aggregate key $\tX$ and is verified by $\ver$.

\paragraph{Correctness}
Fix a quorum size $\mu \geq t$, the number of shares its aggregation algorithms require.
Correctness requires that an honest outer session completes to a valid signature under $\tX$. For every quorum $C$ of at least $\mu$ honest members, $\preagg$ and $\signagg$ produce a first-round and a second-round message that a single signer of the outer scheme holding the secret key of $\pk$ could have produced, and these messages are the same for every qualifying quorum, as \cref{cstr:nonce-fixed} demands.

\paragraph{Security}
We define unforgeability of nested threshold multi-signature schemes through the experiment TS-MS-EUF-CMA, and our notation follows~\cite{Koh26}.
The adversary statically corrupts fewer than $t$ members, controls all signers of the outer session, and controls the network during the interactive key generation of the honest members.
It interacts with the honest members through a key generation oracle $\keygeno$ and two signing oracles $\preroundo$ and $\signo$, which return a member's first-round share for a session identifier $\sid \in \sidspace$ of the adversary's choice, and the member's partial signature for a message and an outer context of the adversary's choice.
Across all honest members, each session identifier is signed on at most one message, modeling the consensus on the live state~(\cref{cstr:consensus}).
The adversary wins if it outputs a fresh forgery, that is a signature that the outer scheme verifies under the aggregate key of a key list containing the group key, on a message $m^*$ such that, for every $\sid$, fewer than $t - \verts{\corrupt}$ honest members completed the second round.
We define $\advantage{\tsmseufcma}{\adv,\Sigma} = \Pr\left[\pcgame{\tsmseufcma}{\Sigma} = \tr\right]$ and provide the formal experiment in \cref{fig:game:tsms}. Note that since we assume that $\binom{n}{t}$ is small by \cref{cstr:small}, an adversary that fully adaptively corrupts signers can be reduced to a static security by a standard guessing argument.

\section{Our Scheme} \label{sec:construction}
We describe our nested threshold multi-signature scheme $\scheme{1}$. 
For a high-level overview, we refer to \cref{sec:to} and to \cref{fig:iceberg1} for pseudocode.

\begin{wfigure}
 \centering
 \fitbox{\linewidth}{%
 \begin{varwidth}{3\linewidth}
 \begin{pcvstack}[boxed,center]
  \pcsetargs{linenumbering}
  \begin{pchstack}
   \begin{pcvstack}
     \procedure{$\setup(\secparam)$}{
       \rule{0pt}{\baselineskip}
       \param \defeq \gparam \gets \grgen(\secparam)\\
       \pcreturn \param
     }
     \pcvspace
     \procedure{$\keygen(k, n, t, A)$}{
       \rule{0pt}{\baselineskip}
       \pcassert n\geq 2t-1\\
       \sk_k \gets \vpss{1}.\keygen(n, t, 2t-1)[k]\\
       (x_k, X_k)\defeq \vpss{1}.\gen(k, \sk_k, w_0)\\
       (C, \{X_j\}_{j\in C}) \gets A(X_k)\\
       \pcassert \vpssver{1}(t, 2t-1, C, \{X_j\}_{j\in C}) = 1\\
       X \defeq \vpss{1}.\agg(t, 2t-1, C, \{X_j\}_{j\in C})\\
       \pcreturn (X, \sk_k)
     }
   \end{pcvstack}
   \pchspace
   \begin{pcvstack}
     \procedure{$\preround(k, \sk_k, \sid)$}{
       \rule{0pt}{\baselineskip}
       \pclinecomment{$\sid$ is assumed to have a fixed length}\\
       \pcfor i\in\interval{\nu}\pcdo\\
       \t (r_{i,k}, R_{i,k}) \defeq \vpss{1}.\gen(k, \sk_k, i\concat \sid)\\
       \pcreturn (R_{1,k},\ldots, R_{\nu, k})
     }
     \pcvspace
     \procedure{$\preagg(\pk, C, (R_{i,j})_{i\in\interval{\nu},j\in C})$}{
       \rule{0pt}{\baselineskip}
       \pcfor i\in\interval{\nu}\pcdo\\
       \t \pcassert \vpssver{1}(t, 2t-1, C, \{R_{i,j}\}_{j\in C}) = 1\\
       \t R_i' \defeq \vpss{1}.\agg(t, 2t-1, C, \{R_{i,j}\}_{j\in C})\\
       b_1 \defeq \Hnon(\pk, (R_1',\ldots, R_\nu'))\\
       \pcfor i \in \interval{\nu}\pcdo\\
       \t R_i \defeq (R_i')^{b_1^{i-1}}\\
       \pcreturn (R_1, \ldots, R_\nu)
     }
   \end{pcvstack}
   \pchspace
   \begin{pcvstack}
     \procedure{$\signround(k, \sk_k, \pk, \sid, m, (R_i')_{i\in\interval{\nu}}, \overline{j}, \{\overline{\pk}_i\}_{i\in\interval{\overline{n}}\setminus \{\overline{j}\}}, (\overline{R}_i)_{i\in\interval{\nu}})$}{
       \rule{0pt}{\baselineskip}
       (x_k, X_k) \defeq \vpss{1}.\gen(k, \sk_k, w_0)\\
       \pcfor i\in\interval{\nu}\pcdo: \;
	 (r_{i,k}, R_{i,k}) \defeq \vpss{1}.\gen(k, \sk_k, i\concat \sid)\\
       \overline{\pk}_{\overline{j}} \defeq \pk\pcsc L\defeq (\overline{\pk}_j)_{j\in\interval{\overline{n}}}\\
       \tX\defeq \keyagg(L) \pcsc a \defeq \keyaggcoef(L, \pk)\\
       b_1 \defeq \Hnon(\pk, (R_1',\ldots, R_\nu'))\pcsc b_0 \defeq \Hnonover(\tX, (\overline{R}_1,\ldots, \overline{R}_\nu), m)\\
       \check{b} \defeq b_1\cdot b_0\pcsc R \defeq \sprod_{i\in\interval{\nu}}\overline{R}_i^{b_0^{i-1}}\pcsc c \defeq \Hsig(\tX, R, m)\\
       s_k \defeq \ssum_{i\in\interval{\nu}}r_{i,k}\,\check{b}^{i-1} + c\cdot a\cdot x_k\\
       \pcreturn s_k
     }
     \pcvspace
     \procedure{$\signagg(R, C, \{s_k\}_{k\in C})$}{
       \rule{0pt}{\baselineskip}
       \pclinecomment{$R$ can be computed from protocol messages as in $\signround$}\\
       \pcfor k\in C\pcdo: \; \lambda_k \defeq \sprod_{j\in C\setminus\{k\}}\frac{j}{j - k}\\
       s \defeq \ssum_{k\in C} s_k\cdot \lambda_k\\
       \pcreturn (R, s)
     }
   \end{pcvstack}
  \end{pchstack}
 \end{pcvstack}
 \end{varwidth}}
 \caption{The protocol $\scheme{1}[\grgen, \nu, t]$. In $\signround$, the values $(R_i')_{i\in\interval{\nu}}$ are the aggregate pre-nonces computed in $\preagg$, and $(\overline{R}_i)_{i\in\interval{\nu}}$ are the aggregate nonces of the outer session.}
 \label{fig:iceberg1}
\end{wfigure}

\paragraph{Setup and key generation}
To set up $\scheme{1}$ we select a group description $\gparam$, and the hash functions $\Hprf, \Hagg, \Hnon, \Hnonover, \Hsig:\str\to\Zpp$ are those of the outer session.
Session identifiers are bitstrings of a fixed length, $\sidspace = \bool^{\ell}$, and the key tag $w_0$ is chosen outside $\sidspace$.
Key generation wraps the key generation of $\vpss{1}$, where the members obtain replicated shares $\sk_k$ of a secret $\sk = \ssum_{\aset\in\asets}\phi_{\aset}$ with quorum parameter $\mu = 2t-1$, from a trusted dealer or a distributed key generation.
Each member derives its key share $(x_k, X_k) \defeq \vpss{1}.\gen(k, \sk_k, w_0)$ at a fixed public tag $w_0$ reserved for the key, the share commitments of a quorum $C$ are checked with $\vpssver{1}$, and $\agg$ interpolates the group key $X = g^{\sk}$.
This key occupies the group's slot of the outer key list, so to the outer session the group looks like one ordinary signer.

\paragraph{Nonce precomputation}
The first round runs before the message exists and fixes the group's nonce for the session identifier $\sid$.
For each nonce slot $i \in \interval{\nu}$, member $k$ derives $(r_{i,k}, R_{i,k}) \defeq \vpss{1}.\gen(k, \sk_k, i\concat\sid)$ and broadcasts the group elements.
$\preagg$ checks the shares of a quorum $C$ with $\vpssver{1}$, interpolates each slot to the pre-nonce $R_i'$, derives the binding factor $b_1 \defeq \Hnon(\pk, (R_1',\ldots,R_\nu'))$, and outputs the bound nonces $R_i \defeq (R_i')^{b_1^{i-1}}$.
This mirrors $\signagg$ followed by $\signaggext$ in $\nestedmusig$, so the output is exactly the first-round message a single MuSig2 participant would send.
Since $\gen$ is deterministic and every summand is replicated, every qualifying quorum computes the same nonces, and a member absent in this round recomputes its shares later from $\sk_k$ alone.

\paragraph{Signing}
The second round runs once the message $m$ and the outer context $\ctx$ are known, that is, the group's position $\overline{j}$ in the outer key list, the co-signer keys $\{\overline{\pk}_i\}$, and the outer aggregate nonces $(\overline{R}_i)$.
Member $k$ recomputes its key share and nonce shares via $\gen$, forms the outer key list $L$ with $\pk$ at position $\overline{j}$, and computes the aggregate key $\tX \defeq \keyagg(L)$ together with the group's aggregation coefficient $a \defeq \keyaggcoef(L, \pk)$.
Two binding factors enter the response, the inner $b_1 \defeq \Hnon(\pk, (R_1',\ldots,R_\nu'))$ and the outer $b_0 \defeq \Hnonover(\tX, (\overline{R}_1,\ldots,\overline{R}_\nu), m)$.
With $\check{b} = b_1\cdot b_0$, the session nonce $R = \sprod_i \overline{R}_i^{b_0^{i-1}}$, and the challenge $c = \Hsig(\tX, R, m)$, member $k$ outputs the partial signature
\[
  s_k \defeq \ssum_{i\in\interval{\nu}} r_{i,k}\,\check{b}^{i-1} + c\cdot a\cdot x_k.
\]

\paragraph{Aggregation and verification}
$\signagg$ combines the partial signatures of a quorum $C$ by Lagrange interpolation, $s \defeq \ssum_{k\in C} s_k\cdot\lambda_k$, and outputs $(R, s)$. The value $s$ is the group's second-round message in the outer session.
The outer session adds the partial signatures of the co-signers, and the result verifies as an ordinary Schnorr signature under $\tX$.
For Lightning we instantiate $\nu = 2$ and nesting depth $\Lambda = 2$, so the outer session is unmodified BIP~327 MuSig2 and the counterparty runs the protocol it already runs today.

\section{Security}
\label{sec:security}

We prove the security of $\scheme{1}$ by reducing it to the EUF-CMA security of $\nestedmusig$, in two steps.
First, in \cref{lem:sharconv}, we reduce $\scheme{1}$ to the scheme $\scheme{0}$~(\cref{fig:iceberg0}), which replaces $\vpss{1}$ by the variant $\vpss{0}$ without share conversion~(\cref{fig:vpss0}), where every member outputs one contribution per summand, and aggregation deduplicates the replicated copies instead of interpolating.
Both schemes compute the same keys, nonces, and signatures, so the reduction is a wrapper that converts shares back and forth.
Then, in \cref{thm:main}, we reduce $\scheme{0}$ to $\nestedmusig$, and we describe this reduction at a high level before stating the theorem.

\begin{lemma}\label{lem:sharconv}
If $\scheme{0}$ is $\tsmseufcma$ secure then $\scheme{1}$ is $\tsmseufcma$ secure. In particular, for every adversary $\adv$ against $\pcgame{\tsmseufcma}{\scheme{1}}$ there exists an adversary $\bdv$ such that $$\advantage{\tsmseufcma}{\bdv,\scheme{0}}\geq \advantage{\tsmseufcma}{\adv,\scheme{1}}.$$
\end{lemma}
\begin{proof}
The algorithm $\bdv$ is an elementary wrapper around $\adv$ that simulates the three TS-MS-EUF-CMA oracles by forwarding all oracle queries and performing share conversion, from replicated secret sharing to Shamir secret sharing, as follows.

Since $\vpss{0}$ and $\vpss{1}$ have identical key generation algorithms, $\bdv$ forwards all $\keygeno$ queries from $\adv$ to the $\keygeno$ oracle provided by its environment until that phase of key generation is complete, making sure to store all extracted secret PRF keys. For each queried index $k$, this results in the environment returning $(X_{k,\aset})_{\aset\in\asetsof{k}}$, from which $\bdv$ computes and returns to $\adv$ the value $X_k\defeq \sprod_{\aset\in\asetsof{k}} X_{k,\aset}^{L_{\aset}'(k)}$, where $L_{\aset}'(k) = \sprod_{\ell\in \aset}\frac{\ell-k}{\ell}$ as in $\vpss{1}.\gen$. Upon subsequently receiving $(C, \{X_j\}_{j\in C})$ from $\adv$, $\bdv$ asserts that $\vpss{1}.\ver(t, 2t-1, C, \{X_j\}_{j\in C}) = 1$ and then obtains, for each $j\in C$, the share vector $(X_{j,\aset})_{\aset\in\asetsof{j}}$, either as observed in the interaction with the environment for honest $j$, or by computing $\vpss{0}.\gen(j, \sk_j, w_0)$ from the shares $\sk_j$ extracted during key generation for corrupt $j$. Then $\bdv$ gives $\{(X_{j, \aset})_{\aset\in\asetsof{j}}\}_{j\in C}$ to the environment and receives $X$, which is returned to $\adv$.

Upon receiving a $\preroundo(\sid, k)$ query from $\adv$, $\bdv$ forwards it to the $\scheme{0}$ $\preroundo$ oracle, receiving $(R_{1,k},\ldots, R_{\nu,k})$, where $R_{i,k} = (R_{i,k,\aset})_{\aset\in\asetsof{k}}$. Then, $\bdv$ computes $$R'_{i,k} \defeq \prod_{\aset\in\asetsof{k}} R_{i,k,\aset}^{L_{\aset}'(k)}$$ and returns $(R'_{1,k},\ldots,R'_{\nu,k})$ to $\adv$.

Upon receiving $\signo$ queries from $\adv$, $\bdv$ forwards these queries to the $\scheme{0}$ $\signo$ oracle, receiving $(s_{k,\aset})_{\aset\in\asetsof{k}}$, and then computes and returns to $\adv$ the value $$s_k \defeq \sum_{\aset\in\asetsof{k}}s_{k,\aset}\cdot L'_{\aset}(k).$$

Finally, $\bdv$ outputs the forgery output by $\adv$, which is valid in the $\scheme{0}$ game since both schemes share the same key aggregation and signature verification. The above perfectly simulates the $\scheme{1}$ environment to $\adv$, noting that both $\vpss{0}$ and $\vpss{1}$ are information-theoretically verifiable and unique.
\end{proof}

\paragraph{Reducing to $\nestedmusig$}
The reduction behind \cref{thm:main} is a straight-line simulation. It never rewinds the adversary, and it introduces no assumption beyond the security of $\nestedmusig$, which has multiple security proofs under different parameter constraints and cryptographic assumptions.
Its starting point is the observation that a group running $\scheme{0}$ already behaves like a depth-two $\nestedmusig$ session.
The group key is the product of the summand keys $g^{\Hprf(\phi_{\aset}, w_0)}$, one per subset $\aset$, so we may view the summands themselves as the innermost signers of a nesting, where the outer level is the MuSig2 session with the counterparty, and the inner level aggregates the summand keys into the group key.
The reduction makes this view literal by programming the key aggregation oracle $\Hagg$ so that the aggregate of the summand keys equals the group key, and it hides the EUF-CMA challenge key in exactly one summand, one whose PRF key the corrupted members do not hold.

Every value the adversary sees is then computed honestly, with one exception.
For all summands except the challenge one, the reduction knows the PRF keys and runs the real protocol.
For the challenge summand it cannot evaluate $\Hprf$, so it substitutes the honest signer of the EUF-CMA game. The nonce shares for a session identifier come from the $\preroundo$ oracle, and the corresponding component of a partial signature comes from the $\signo$ oracle.
This substitution is consistent because, in $\scheme{0}$, the component a member reports for a summand depends only on the summand and the session, not on the member, so one oracle response answers every honest member's query for the same $\sid$, and caching it preserves the determinism the adversary expects. Furthermore, the consensus assertion in the game guarantees that each $\sid$ is signed on at most one message, so the single $\sign'$ call that the EUF-CMA game permits per session is all the reduction ever needs.

A forgery translates because if for every $\sid$, fewer than $t - \verts{\corrupt}$ honest members signed $m^*$, then, for a suitable choice of the challenge summand, no component of the challenge summand for $m^*$ was ever released. 
The reduction then never queried $\signo$ on $m^*$, and the forged signature is a valid $\nestedmusig$ forgery on a fresh message, and \cref{lem:sharconv} lifts the result from $\scheme{0}$ to $\scheme{1}$.

\begin{theorem}\label{thm:main}
If the nested multi-signature scheme $\nestedmusig[\grgen, \nu]$ is EUF-CMA secure, then the nested threshold multi-signature scheme $\scheme{1}[\grgen, \nu, t]$ is TS-MS-EUF-CMA secure in the random oracle model for $\Hagg$, $\Hprf$, $\Hnon$, $\Hnonover$, $\Hsig:\str\to\Zpp$.

Precisely, let $N = \binom{n}{t-1}$ be the number of summands of the replicated sharing, which is also the size of the key list the reduction hands to the $\nestedmusig$ session.
For any adversary $\adv$ against $\scheme{1}$ running in time at most~$\tau$ and making at most $q_{agg}$ queries to $\Hagg$, there exists an adversary $\cdv$ against $\nestedmusig$ running in time at most
\[
  \tau' = \tau + O(N)\cdot \tau_{\rm exp},
\]
where $\tau_{\rm exp}$ is the time of an exponentiation in $\GG$, such that
$$
\advantage{\eufcma}{\cdv,\nestedmusig}
\geq
\advantage{\tsmseufcma}{\adv,\scheme{1}} - \frac{N\cdot q_{agg}}{2^\secpar}.
$$
\end{theorem}
\begin{proof}
We proceed by constructing an elementary wrapper, $\cdv$, around $\bdv$ from \cref{lem:sharconv} that simulates the TS-MS-EUF-CMA environment given access to the EUF-CMA environment for $\nestedmusig$. The result then follows since, as we argue below, a forgery in the simulated game translates into an EUF-CMA forgery against $\nestedmusig$.

First, $\cdv$ performs $\vpss{0}.\keygen$ simulating all honest parties, allowing it to extract all adversarial PRF secrets by the honest majority assumption, which assures that the honest parties make up a valid quorum. Index the subsets in $\asets$ as $\aset_1,\ldots,\aset_\gamma$ with $\gamma = \verts{\asets}$ and, WLOG, let $\phi_{\aset_1}$ be a PRF key not known to the adversary at the end of secret sharing.

Next, $\cdv$ embeds the EUF-CMA challenge key $X^*$ by setting $X_1 \defeq X^*$, so that the corresponding summand key becomes $(X^*)^{a_1}$ for the programmed coefficient $a_1$ below. Specifically, for each $j>1$, $\cdv$ computes $x_j' \defeq \Hprf(\phi_{\aset_j}, w_0)$, picks $a_j$ randomly, and then sets $x_j \defeq x_j'/a_j$ and $X_j \defeq g^{x_j}$. Finally, $\cdv$ programs $\Hagg(L, X_j)$ to be equal to $a_j$ for all $j$, where $L \defeq \{X_1,\ldots,X_\gamma\}$. This ensures that $\keyagg(L) = \sprod_j X_j^{a_j} = (X^*)^{a_1}\sprod_{j>1} g^{x_j'}$. The programming fails with only negligible probability since the value $X_1$ was just drawn uniformly at random and is an element of $L$. Concretely, if $q_{agg}$ is a bound on the number of queries the adversary makes to $\Hagg$ then the programming fails with probability at most $\frac{N\cdot q_{agg}}{2^\secpar}$. Lastly, $\cdv$ continues to execute key generation honestly, except that $(X^*)^{a_1}$ is used in place of $g^{\Hprf(\phi_{\aset_1}, w_0)}$.

Upon receiving $\preroundo$ queries, $\cdv$ honestly makes all expected $\sid$-related assertions and computes $\preround$ for all but the unknown summand $\aset_1$ in which the challenge key is embedded. For that summand, $\cdv$ calls the EUF-CMA $\preroundo$ oracle. All future calls with the same session identifier use a cached response for the challenge summand to maintain consistency. Since we are assuming state-agreement is being used so that it is impossible to validly query the same sid in two different inputs to the second round, this simulation is faithful.

Lastly, upon receiving $\signo$ queries, $\cdv$ honestly makes all verification assertions and then computes $\signround$ honestly for all but the unknown summand $\aset_1$. For that summand, $\cdv$ invokes the EUF-CMA $\signo$ oracle.
Because of the programming of $\Hagg$ during key generation, the response from the nested MuSig2 signing oracle yields a valid $\scheme{0}$ partial signature.

We conclude that $\cdv$ perfectly simulates the expected environment, so long as no bad event happens in $\Hagg$, since everything is computed honestly except for the values relating to the embedded challenge, all of which are drawn from a distribution identical to an honest execution.
\end{proof}

\begin{corollary}
Let $\grgen$ be a group generation algorithm for which the algebraic one-more discrete logarithm problem is hard.
By \cref{thm:main} and the security of $\nestedmusig[\grgen, \nu=2]$ for sessions of depth two~\cite{Koh26}, the scheme $\scheme{1}[\grgen, \nu=2, t]$ nested in an unmodified BIP~327 session is TS-MS-EUF-CMA secure in the algebraic group model for $\grgen$ and the random oracle model for $\Hagg$, $\Hprf$, $\Hnon$, $\Hnonover$, $\Hsig$.
\end{corollary}

\section{Prototype and Evaluation} \label{sec:evaluation}

\definecolor{cbBlue}{HTML}{0072B2}
\definecolor{cbOrange}{HTML}{E69F00}
\definecolor{cbGreen}{HTML}{009E73}
\definecolor{cbVermillion}{HTML}{D55E00}
\pgfplotsset{
  icebergplot/.style={
    tick align=outside, tick pos=left,
    grid style={very thin, black!12},
    axis line style={black!55}, tick style={black!55},
    minor tick length=1pt, major tick length=2.5pt,
    every axis plot/.append style={semithick},
    label style={font=\footnotesize},
    tick label style={font=\scriptsize},
    legend cell align=left,
    legend style={font=\scriptsize, draw=none, fill=none},
    mark size=1.7pt,
  },
}

We support our work with a prototype and a performance evaluation.
Our prototype~\prototypecite~implements \Iceberg over \texttt{secp256k1} and ports it into \texttt{eclair}, a production Lightning implementation.
Beyond demonstrating that \Iceberg can thresholdize today's Lightning with only moderate overhead, we structure our evaluation around four questions.
\begin{enumerate}[itemsep=2pt]
\item \emph{Micro benchmarks.} How expensive is \Iceberg in computation and in wire bytes across $(\thres,\groupsize)$?
\item \emph{Drop-in replacement.} Does \Iceberg slot into a production Lightning node without changes to its channel logic?
\item \emph{Thresholdized Lightning throughput.} How many payments per second can an \Iceberg Lightning endpoint sustain on a single core?
\item \emph{System bottleneck.} Does the threshold signature limit a thresholdized channel, or does the channel logic dominate it?
\end{enumerate}

\paragraph{Experimental setup}
We implement \Iceberg over \texttt{secp256k1} with BIP~340 x-only keys~\cite{BIP340} and BIP~327 key aggregation~\cite{BIP327}.
The outer session is unmodified BIP~327 with $\nu = 2$, so the counterparty sees stock MuSig2 throughout.
All measurements run on a dedicated machine with an AMD EPYC 7543 CPU, ten cores made available to the benchmark and a $4$\,GB measurement heap.
We measure at two layers, namely the primitive in isolation and a whole payment driven through two \texttt{eclair} channel state machines.
Both run over the parameter family $\thres \in \{2,3,4,5\}$ with $\groupsize$ up to ten, resulting in twenty configurations.\footnote{
For every $\thres$, the signing quorum has at least $2\thres - 1$ members.}

\paragraph{Micro benchmarks}
We first measure the local cost of \Iceberg and compare it to single-key MuSig2.  
One ordinary MuSig2 participant spends $42.0$\,$\mu$s on its partial signature, with a standard deviation of $5.4$ across all twenty configurations, measured over $200$ iterations. 
A group instead spends $1\,042$ to $7\,358$\,$\mu$s on the same partial signature~(\cref{tab:micro}).
The reported numbers are totals across the group, since every member runs in one process here.
Spread over the members at $(3, 7)$, a member that signs in both rounds does $14.4 \times$ the work of one participant, and a member that only contributes to the first round $1.5 \times$.
Aggregation costs additional $392$\,$\mu$s per signature.


On the wire the group still acts as a single signer, sending the counterparty $66$ bytes in the first round and $32$ in the second, in accordance with BIP~327. 
The traffic inside the group costs $2\,208$ bytes per payment at $(3, 7)$, which is $2\thres-1$ nonce shares in each of six first rounds and $\thres$ signature shares in each of two second rounds.
Each share carries a one-byte member index on top of the $66$ or $32$ bytes it transports.
Beyond that traffic, a member stores $100$ bytes of long-lived share material at $(2, 4)$, $484$ at $(3, 7)$ and $2\,692$ at $(4, 10)$.

\begin{table}[tp]
\tabcaption{Total processor time the group spends on one partial signature, in $\mu$s, measured in C over $200$ iterations. Empty cells cannot exist, since the quorum of $2\thres-1$ has to fit inside the group.}{tab:micro}

\centering
\tabfont
\setlength{\tabcolsep}{3.4pt}
\begin{tabular}{l r r r r r r r r}
\toprule
 & \multicolumn{8}{c}{Group size $\groupsize$} \\
\cmidrule(lr){2-9}
$\thres$ & $3$ & $4$ & $5$ & $6$ & $7$ & $8$ & $9$ & $10$ \\
\midrule
$2$ & $1\,042$ & $1\,226$ & $1\,242$ & $1\,099$ & $1\,056$ & $1\,083$ & $1\,087$ & $1\,072$ \\
$3$ & & & $2\,138$ & $2\,156$ & $2\,214$ & $2\,287$ & $2\,270$ & $2\,441$ \\
$4$ & & & & & $3\,685$ & $3\,841$ & $4\,180$ & $4\,557$ \\
$5$ & & & & & & & $6\,240$ & $7\,358$ \\
\bottomrule
\end{tabular}
\tabcaptionend
\end{table}

\paragraph{Drop-in replacement}
Next, we demonstrate that \Iceberg is a drop-in replacement for MuSig2. 
\texttt{eclair} reaches MuSig2 through the \texttt{secp256k1-kmp} bindings, and every call it makes for a taproot channel goes through that interface. 
We implement \Iceberg behind that interface so \texttt{eclair} obtains a public key, a public nonce and a partial signature from the group wherever it previously obtained them from one signer. 
Therefore, we can use \Iceberg in \texttt{eclair} without changing the channel logic or the message flow. 

Next, we verified that our binding is valid with the checks included by \texttt{eclair}. 
A channel enters its operating state only after \texttt{eclair} rebuilds the commitment transaction and runs it through Bitcoin's script interpreter against the real funding output, on both sides.
Every payment we time therefore settled against a genuine aggregate Schnorr signature.
We also replay the whole path against complete signing sessions recorded from the reference implementation.
Given the same seeds, session identifier and message, our prototype reproduces the reference byte for byte at each member's public share, at the group public key, at the aggregate group nonce and at the group signature.
We run this replay at $(2,3)$, $(3,5)$, $(3,7)$, $(4,7)$ and $(5,9)$.
Beyond that we run ten sequential payments on one channel.
We also pin both orderings of the two public keys, since key aggregation sorts them and the group lands first about half the time.
Finally, we run the group as the party that opens the channel and as the party that accepts it.
None of this changes \texttt{eclair}'s own behaviour, since $473$ of its channel tests pass unchanged and the build that produced our timings passed all seventeen of its correctness checks.

We additionally verified that a \Iceberg group never revises the nonce it sent even with a changing signing quorum~(\cref{cstr:offline,cstr:nonce-fixed}).
We verify that \Iceberg is compliant with these constraints by completing a session in which the members that produce the second round did not contribute to the first.
The resulting signature verifies under the aggregate key at every configuration we replay.

\paragraph{Thresholdized Lightning throughput}
We then measure the sustained throughput an \texttt{eclair} channel achieves once \Iceberg thresholdizes one of its two endpoints~(\cref{tab:payment}). 
Making one endpoint a group costs $3\,844$\,$\mu$s per payment at $(2,4)$ and $8\,246$\,$\mu$s at $(3,7)$, which is $6.7\%$ and $14.3\%$ of a payment's CPU time.
The largest configuration we deployed, $(4,10)$, adds $16\,786$\,$\mu$s, yet all configurations retain substantial throughput headroom over current Lightning payment rates.
Across the deployable family the growth is therefore quadratic. 

The added work is fixed per payment, so it converts into throughput on one core~(\cref{tab:price}).
A bare channel sustains $17.4$ payments per second, a $(2,4)$ group $16.3$ and a $(3,7)$ group $15.2$.
Matching the bare rate therefore takes $6.7$ to $29.1$ percent more machines, and $9.9$ to $38.0$ percent if the members verify each other's partial signatures.

Those three configurations pin $\groupsize$ to $\thres$, since a deployable group has $3\thres-2$ members, so they cannot show which of the two drives the cost.
We therefore sweep the whole family and ask whether the cost follows the fault tolerance alone or also the number of key shares a member holds~(\cref{tab:sweep}).
A member holds $\binom{\groupsize-1}{\thres-1}$ key shares, which is two to nine along the $\thres = 2$ row and twenty to eighty-four along the $\thres = 4$ row.
The $\thres = 2$ row is flat, because its seven added shares move the cost by $211$\,$\mu$s against intervals of about $\pm 160$.
The $\thres = 4$ row rises by $4\,258$\,$\mu$s instead, because its sixty-four added shares cost roughly $67$\,$\mu$s each.
The fault tolerance therefore sets the bulk of the cost, and the shares add to it only once the threshold makes them numerous.
The entry marked \textit{n/r} is not resolved, because its interval spans zero.

\begin{table}[tp]
\tabcaption{Added cost of one payment at the three deployable configurations, against a pooled control of $57\,599$\,$\mu$s.}{tab:payment}
\centering
\tabfont
\begin{tabular}{l l r r r}
\toprule
Tolerates & Configuration & Quorum & Added ($\mu$s) & Relative \\
\midrule
$1$ faulty member  & $(2,4)$  & $3$ & $3\,844 \pm 157$ & $6.7\%$ \\
$2$ faulty members & $(3,7)$  & $5$ & $8\,246 \pm 202$ & $14.3\%$ \\
$3$ faulty members & $(4,10)$ & $7$ & $16\,786 \pm 247$ & $29.1\%$ \\
\bottomrule
\end{tabular}
\tabcaptionend
\end{table}

\begin{table}[tp]
\tabcaption{Payment throughput on one core, derived from the measured cost per payment. The last column adds the optional verification of a member's partial signature.}{tab:price}
\centering
\tabfont
\begin{tabular}{l r r r}
\toprule
 & Payments & Against & With member \\
Configuration & per second & bare MuSig2 & verification \\
\midrule
bare MuSig2 & $17.4$ & & \\
$(2,4)$  & $16.3$ & $-6.3\%$  & $15.8$ $(-9.0\%)$ \\
$(3,7)$  & $15.2$ & $-12.5\%$ & $14.5$ $(-16.4\%)$ \\
$(4,10)$ & $13.4$ & $-22.6\%$ & $12.6$ $(-27.5\%)$ \\
\bottomrule
\end{tabular}
\tabcaptionend
\end{table}

\begin{table}[tp]
\tabcaption{Added cost of one payment in $\mu$s, over $1\,500$ paired iterations per configuration.}{tab:sweep}
\centering
\tabfont
\setlength{\tabcolsep}{3.4pt}
\begin{tabular}{l r r r r r r r r}
\toprule
 & \multicolumn{8}{c}{Group size $\groupsize$} \\
\cmidrule(lr){2-9}
$\thres$ & $3$ & $4$ & $5$ & $6$ & $7$ & $8$ & $9$ & $10$ \\
\midrule
$2$ & $3\,884$ & $3\,844$ & $3\,618$ & $3\,960$ & $3\,954$ & \textit{n/r} & $3\,910$ & $4\,096$ \\
$3$ & & & $7\,482$ & $7\,954$ & $8\,246$ & $8\,271$ & $8\,746$ & $8\,836$ \\
$4$ & & & & & $12\,528$ & $13\,535$ & $15\,185$ & $16\,786$ \\
$5$ & & & & & & & $22\,573$ & $26\,124$ \\
\bottomrule
\end{tabular}
\tabcaptionend
\end{table}

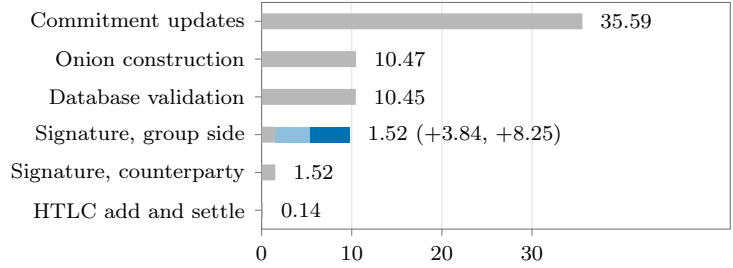
\begin{figure}[!ht]
\centering
\fitbox{\columnwidth}{%
\begin{tikzpicture}
\begin{axis}[icebergplot,
  width=\figwidth, height=3cm, scale only axis,
  xbar stacked, bar width=6pt,
  xmin=0, xmax=52, xtick={0,10,20,30},
  symbolic y coords={commit,onion,db,sig,peer,htlc}, ytick=data, y dir=reverse,
  yticklabels={{Commitment updates},{Onion construction},{Database validation},
               {Signature, group side},{Signature, counterparty},{HTLC add and settle}},
  enlarge y limits=0.1,
  xmajorgrids=true, ymajorgrids=false,
]
\addplot[fill=black!28, draw=none] coordinates
  {(35.59,commit) (10.47,onion) (10.45,db) (1.52,sig) (1.52,peer) (0.14,htlc)};
\addplot[fill=cbBlue!45, draw=none] coordinates
  {(0,commit) (0,onion) (0,db) (3.84,sig) (0,peer) (0,htlc)};
\addplot[fill=cbBlue, draw=none] coordinates
  {(0,commit) (0,onion) (0,db) (4.40,sig) (0,peer) (0,htlc)};
\node[anchor=west, font=\scriptsize, xshift=3pt] at (axis cs:35.59,commit) {$35.59$};
\node[anchor=west, font=\scriptsize, xshift=3pt] at (axis cs:10.47,onion)  {$10.47$};
\node[anchor=west, font=\scriptsize, xshift=3pt] at (axis cs:10.45,db)     {$10.45$};
\node[anchor=west, font=\scriptsize, xshift=3pt] at (axis cs:9.76,sig)     {$1.52$ ($+3.84$, $+8.25$)};
\node[anchor=west, font=\scriptsize, xshift=3pt] at (axis cs:1.52,peer)    {$1.52$};
\node[anchor=west, font=\scriptsize, xshift=3pt] at (axis cs:0.14,htlc)    {$0.14$};
\end{axis}
\end{tikzpicture}}
\caption{Computation splits of a single payment, over $500$ iterations on a warm JVM. The two appended segments are what a $(2,4)$ and a $(3,7)$ group adds to the single signature it replaces.}
\label{fig:components}
\end{figure}

\paragraph{System bottleneck}
Our micro benchmarks show that at $(3,7)$ a group signs $53$ times more slowly than one participant.
At the same time, a payment's compute grows by only $14.3\%$.
Therefore, we ask what the system bottleneck in a thresholdized channel is. 
%
To do this, we decompose a payment into the work \texttt{eclair} performs and compare it to the group's contribution~(\cref{fig:components}).
A payment costs $59.7$\,ms, of which the two commitment updates are $59.6\%$, the sender's onion construction is $17.5\%$ and the test database validating its own writes is another $17.5\%$.
The funding signatures of both parties are $5.1\%$, and only half of that, on average $1\,517$\,$\mu$s, belongs to the party a group replaces.
The substitution therefore starts from a fortieth of a payment.
The commitment updates carry most of that, and inside them the MuSig2 signing takes microseconds.
The rest is therefore transaction construction, one signature per HTLC, per-commitment key derivation, state persistence and message passing between the two channel state machines.

We also time the group's cryptography and the single-key signer it replaces inside the same payment, and hold their difference against the slowdown the channel shows~(\cref{tab:twolayer}).
The two agree to within $350$\,$\mu$s at every deployable configuration, so a group costs no more than the cryptography it adds.
A naive combination of the micro benchmark and the payment decomposition would predict a much larger slowdown.
Scaling a fortieth of a payment by the $53\times$ micro-benchmark factor suggests a payment should more than double in time.
However, the measured growth is only $14.3\%$.
This is because, within the live payment flow, the group costs $6.5\times$ the signer it replaces rather than $53\times$.
We attribute this gap to the fixed overhead of the \texttt{secp256k1-kmp} bindings, which affects both signers.

\begin{table}[tp]
\tabcaption{Cryptography timed inside a live payment against the cost measured through the channel.}{tab:twolayer}

\centering
\tabfont
\begin{tabular}{l r r r r}
\toprule
 & Group & Displaced & Measured & Residual \\
Configuration & ($\mu$s) & signer ($\mu$s) & added ($\mu$s) & ($\mu$s) \\
\midrule
$(2, 4)$  & $5\,672$  & $1\,633$ & $3\,844$  & $-195$ \\
$(3, 7)$  & $9\,946$  & $1\,539$ & $8\,246$  & $-161$ \\
$(4, 10)$ & $18\,557$ & $1\,421$ & $16\,786$ & $-350$ \\
\bottomrule
\end{tabular}
\tabcaptionend
\end{table}

\paragraph{Scope and limitations}
We deliberately exclude the network between the members from our benchmarks.
Since a group is the devices of one user or the machines of one company, we assume co-location on a low-latency local network.
A payment still drives six first rounds and two second rounds inside the group, so a deployment adds one round trip for each of them.
Furthermore, our prototype covers the active payment path alone, and mutual close, force close, splicing and channel announcements are not yet wired to a threshold signer.
Channel announcements in particular need an ECDSA signature over the funding key, which \Iceberg does not produce.
Finally, our benchmarks evaluate payments carrying a single HTLC.
Because every key on the group's side is a group key and each HTLC needs its own signature, the cost grows with the number of HTLCs in flight.

\section{Conclusion and Future Work} \label{sec:conclusion}
\Iceberg allows a group of signers to jointly operate a single participant in a MuSig2 session such that the other participant sees the message flow of a single signer throughout.
Either endpoint of a Lightning channel, or both independently, can use threshold custody with no change to Bitcoin, the Lightning protocol, or its counterparty.
We prove \Iceberg secure and demonstrate its practicality through an integration with \texttt{eclair}: a 2-of-4 endpoint sustains over 93\% of the theoretical maximum payment throughput of an unmodified endpoint, leaving substantial headroom over current Lightning payment rates.

More broadly, \Iceberg shows that thresholdization can be achieved by nesting a threshold group inside an existing multi-signature protocol rather than replacing it.
An open question is which other multi-signature protocols admit such nesting and what properties are required to support it.
For Lightning, a further question is how much a group can change about itself while the channel stays open: proactively refreshing its shares, repairing the share of a member that has lost it, or changing the membership or the threshold.
Each of these has to preserve the aggregate key, since that key is the channel's funding output and changing it means closing and reopening the channel.

\ifnum\fullversion=0
\section*{Ethics Considerations}

This paper proposes a nested threshold multi-signature scheme, proves it secure, and evaluates a prototype that integrates with an existing production Lightning implementation.
We consider the ethical implications of conducting this research and of publishing its results, following a stakeholder-based analysis.

\paragraph{Stakeholders}
The primary stakeholders are the operators of Lightning channel endpoints, from individuals who hold their own keys to exchanges that settle customer deposits over Lightning, and the users whose funds those endpoints hold.
Further stakeholders are the maintainers of the Lightning implementations and of the Bitcoin libraries our prototype builds on, and the wider Bitcoin community that decides which channel types to adopt.
Indirect stakeholders are adversaries who might study the work to understand the defences of a thresholdized endpoint.

\paragraph{Potential benefits}
A Lightning channel holds real funds under a single online key, and both compromise and loss of that key are unrecoverable.
Our work removes that single point of failure without asking the counterparty, the protocol or the chain to change, so the protection can be adopted by one side alone.
The benefit falls to the party that adopts it and does not depend on anyone else acting.

\paragraph{Potential harms and dual use}
Our construction is defensive.
It introduces no attack techniques and it exploits no weaknesses of Bitcoin, MuSig2 or the Lightning protocol.
The one failure mode we analyse, the leakage of a signing key when one nonce is used for two messages, is a textbook property of Schnorr signatures and is the reason our third design decision exists.
We assess the risk that this work enables concrete harm as low, and as outweighed by the benefit of removing an online single point of failure from a system that holds hundreds of millions of dollars.

\paragraph{Research methodology and data handling}
Our evaluation runs a prototype in a controlled environment.
Both channel state machines run on our own machines against an in-memory database, and no channel we measure is opened on mainnet, carries real funds, or connects to a third party's node.
No data from real users is collected, processed or analysed, so the research involves no human subjects, no personal data, and no experiment on a live system.

\paragraph{Decision to publish}
We conclude that publication is justified.
The work strengthens the custody of a system that is already carrying institutional volumes of value, and it introduces no new offensive capability.
We have tried to state the limits of the prototype clearly enough that no reader mistakes it for a production system.

\fi

\ifnum\fullversion=1
\section*{Acknowledgements}
We would like to thank Adam Everspaugh, Leo Nash, and Jonas Nick for helpful discussions in the early stages of this project. 
Paul Gerhart's research has been supported by the Google PhD Fellowship in Privacy, Safety, and Security.
\fi

\ifnum\fullversion=0
\bibliographystyle{IEEEtran}
\else
\bibliographystyle{alpha}
\fi
\bibliography{references}

\ifnum\fullversion=0
\appendices
\else
\appendix

\fi
\section{Deferred Preliminaries} \label{sec:appdx:prelims}
We recall the EUF-CMA security game and pseudocode of $\nestedmusig$~\cite{Koh26}, its corresponding relevant security results, and the verifiable pseudorandom secret sharing $\vpss{1}$ of Arctic~\cite{KG24}.

All of the keys at the leaves of a cosigner tree must cooperate in generating partial signatures in order for the aggregate multi-signature to be computed, just as is the case for a non-nested multi-signature scheme. Thus, the unforgeability game for nested multi-signature schemes is derived directly from the unforgeability game for security under concurrent signing sessions for two-round multi-signatures. Namely, the adversary is given a challenge $\pk^*$ and concurrent oracle access to $\preround$ and $\signround$, where it may set the message to be signed, cosigner tree in which the honest signer is embedded, and all aggregate first round messages as the input to its $\signround$ oracle, and must produce a cosigner tree containing the challenge along with a forgery for the aggregate key on a message not queried to the $\signround$ oracle under that cosigner tree. This is defined rigorously in \autoref{fig:game:eufcma}.

\begin{figure}[t]
 \centering
 \fitbox{\columnwidth}{%
 \begin{varwidth}{2\columnwidth}
 \begin{pcvstack}[boxed,center]
  \pcsetargs{linenumbering}
  \procedure{Game $\pcgame{\eufcma}{\Sigma}$}{%
   \rule{0pt}{\baselineskip}
   \param \gets \setup(\secparam) \\
   \pclinecomment{honest signer has index 1} \\
   (\sk^*,\pk^*) \gets \keygen() \\
   \ctrs \defeq 0 \quad \pccomment{session counter} \\
   S \defeq \emptyset \pccomment{set of open signing sessions after $\preroundo$} \\
   Q \defeq \emptyset \quad \pccomment{set of $\signo$ queries} \\
   ((L_d)_{0\leq d < \Lambda}, m,\sigma) \gets \adv^{\preroundo,\signo}(\pk^*) \\
   \pcfor d \defeq \Lambda - 1,\ldots,0  \pcdo\\
   \t \tX_d \defeq \keyagg(L_d)\\
   \pcassert \pk^* \in L_{\Lambda-1} \wedge \forall d<\Lambda-1,\tX_{d+1}\in L_d\\
   \pcreturn ((L_d)_{0\leq d < \Lambda}, m) \notin Q \wedge \ver(\tX_0,m,\sigma)=\tr
  }
  \pcvspace
  \procedure{Oracle $\preroundo()$}{%
    \rule{0pt}{\baselineskip}
    \ctrs \defeq \ctrs+1 \quad \pccomment{increment session counter} \\
    k \defeq \ctrs\pcsc S \defeq S \cup \{k\} \pccomment{open session $k$}\\
   (\msg_1,\state_{1,k}) \gets \sign() \\
   \pcreturn \msg_1
  }
  \pcvspace
  \procedure{Oracle $\signo\left(k,(\msg^d)_{0\leq d < \Lambda},m,\{\pk_{i,d}\}_{2\leq i\leq n_d, 0\leq d < \Lambda}\right)$}{%
   \rule{0pt}{\baselineskip}
   \pcassert k \in S\pcsc S \defeq S\setminus \{k\}\\
   input \gets \left(\state_{1,k},(\msg^d)_{0\leq d<\Lambda}, \sk^*, m,\{\pk_{i,d}\}_{2\leq i\leq n_d, 0\leq d < \Lambda}\right)\\
   (\state_{1,k}',\msg_1') \gets \sign'(input) \\
   L_{\Lambda - 1} \defeq \{\pk^*,\pk_{2,\Lambda-1},\ldots,\pk_{n_{\Lambda-1},\Lambda-1}\} \\
   \pcfor d\defeq \Lambda - 2,\ldots, 0\pcdo\\
   \t \pk_{1,d} \defeq \keyagg(L_{d+1})\\
   \t L_d \defeq \{\pk_{1,d}, \ldots, \pk_{n_d, d}\}\\
   Q \defeq Q \cup \{((L_d)_{0\leq d < \Lambda}, m)\}\\
   \pcreturn \msg_1'
  }
 \end{pcvstack}
 \end{varwidth}}
 \caption{The EUF-CMA security game for a two-round nested multi-signature scheme $\Sigma$, adapted from~\cite{Koh26}.}
 \label{fig:game:eufcma}
\end{figure}

The $\eufcma$ security of $\nestedmusig$, as defined in \autoref{fig:game:eufcma}, is exactly what we reduce the security of $\scheme{1}$ to in this paper. Under this notion of security, $\nestedmusig$ has two security results from~\cite{Koh26}, one proven in the ROM and the other in the ROM+AGM.

\begin{theorem}\label{thm:musigrom}
Let $\grgen$ be a group generation algorithm for which the AOMDL problem is hard.
The nested multi-signature scheme $\mathsf{NestedMuSig2}[\grgen,\nu=8, D=2]$ is EUF-CMA in the random oracle model for $\Hagg$, $\Hnon$, $\Hnonover$, $\Hsig:\str\to\Zpp$.

Precisely, for any adversary $\adv$ against $\mathsf{NestedMuSig2}[\grgen,\nu=8]$ running in time at most $\tau$, making at most $q_s$ $\signo$ queries and at most $q_h$ queries to each random oracle, and such that the size of $L_d$ in any signing session and in the forgery is at most $N$, and the maximum depth of any signing session is at most $2$, there exists an algorithm $\edv$
  taking as input group parameters $\gparam \gets \grgen(\secparam)$,
  running in time at most
\[
  \tau'=8(\tau + q(3N + 3\nu - 3))\tau_{\rm exp} + O(qN)
\]
 where $q=(N+4)q_h + (2N + 3)q_s + 2N + 1$ and $\tau_{\rm exp}$ is the time of an exponentiation in $\GG$,
 making at most $8q_s$ $\DL_g$ queries,
 and solving the AOMDL problem with an advantage
\[
  \advantage{\aomdl}{\edv,\grgen}
  \ge \frac{\epsilon^8}{16q^7} - \frac{16q + 2\binom{8q + 2}{2}^2 + 12}{2^\secpar},
\]
 where $\epsilon = \advantage{\eufcma}{\adv,\mathsf{NestedMuSig2}[\grgen,\nu=8, D=2]}$.
\end{theorem}

\begin{theorem}\label{thm:musigagm}
Let $\grgen$ be a group generation algorithm for which the AOMDL problem is hard.
Then the multi-signature scheme $\mathsf{NestedMuSig2}[\grgen,\nu=2, D=2]$ is EUF-CMA in the algebraic group model for $\grgen$, and the random oracle model for $\Hagg$, $\Hnon$, $\Hnonover$, $\Hsig:\str\to\Zpp$.

Precisely, for any algebraic adversary $\adv$ against $\mathsf{NestedMuSig2}[\grgen,\nu=2]$ running in time at most~$\tau$, making at most $q_s$ $\signo$ queries, at most $q_h$ queries to each random oracle, such that the size of $L$ in any signing session and in the forgery is at most $N$, and such that the depth, $\Lambda$, of any signing session or forgery is at most $2$, there exists an algorithm $\bdv$ running in time at most
\[
  \tau'= \tau + O(qN)\cdot \tau_{\rm exp} + O(q^2),
\]
where $q = \binom{2q_h + 2q_s + 2}{2} + 2N(q_s + 1) + Nq_h$ and $\tau_{\rm exp}$ is the time of an exponentiation in $\GG$ and making at most $2 q_s$ $\DL_g$ queries such that
\[
 \advantage{\aomdl}{\bdv,\grgen}
 \geq
 \advantage{\eufcma}{\adv,\mathsf{NestedMuSig2}[\grgen,\nu=2]}
 - \frac{30q^4}{2^\secpar}.
\]
\end{theorem}

Lastly, in \autoref{fig:vpss1} below, we recall the definition of the Verifiable Pseudorandom Secret Sharing function $\vpss{1}$ as defined in~\cite{KG24}. This protocol allows a group of $n$ parties with an honest majority to perform replicated secret sharing on pseudorandom function seeds to generate an arbitrary number of (replicated) shared secrets. Share conversion (introduced in~\cite{TCC:CraDamIsh05}) is then used to turn these replicated secret shares into Shamir secret shares locally without any communication required. In the honest-majority setting, $\vpss{1}$ allows parties to verify function outputs in aggregate by performing verification ``in the exponent'' and using the fact that an honest majority implies that any honest set of $t$ parties fully determine an interpolated curve of degree $t-1$ so that interpolating a higher degree curve proves dishonest activity has occured.
\section{Additional Figures} \label{sec:appdx:figures}

\pairfigstart

  \fitbox{\columnwidth}{%
  \begin{varwidth}{2\columnwidth}
  \begin{pchstack}[boxed,center]
   \pcsetargs{linenumbering}
   \begin{pcvstack}
     \procedure{$\setup(\secparam)$}{
       \rule{0pt}{\baselineskip}
       \param \defeq \gparam \gets \grgen(\secparam)\\
       \pcreturn \param
     }
     \pcvspace
     \procedure{$\keygen(n, t, \mu)$}{
       \rule{0pt}{\baselineskip}
       \pcassert \mu\geq 2t-1\\
       \asets \defeq \sbinom{\interval{n}}{t-1}\\
       \pcfor \aset\in\asets \pcdo\\
       \t \phi_{\aset} \sample \Zpp \pccomment{$\sk = \ssum_{\aset\in\asets}\phi_{\aset}$}\\
       \pcfor k\in \interval{n} \pcdo\\
       \t \sk_k \defeq \{(\aset, \phi_{\aset})\}_{\aset\in\asetsof{k}}\\
       \pcreturn (\sk_1,\ldots, \sk_n) \pccomment{$\verts{\sk_k} = \sbinom{n-1}{t-1}$}
     }
     \pcvspace
     \procedure{$\gen(k, \sk_k, w)$}{
       \rule{0pt}{\baselineskip}
       \{(\aset, \phi_{\aset})\}_{\aset\in\asetsof{k}} \defeq \sk_k\\
       d_k \defeq (d_{k,\aset})_{\aset\in\asetsof{k}} \defeq (\Hprf(\phi_{\aset}, w))_{\aset\in\asetsof{k}}\\
       D_k \defeq (D_{k,\aset})_{\aset\in\asetsof{k}} \defeq (g^{d_{k,\aset}})_{\aset\in\asetsof{k}}\\
       \pcreturn (d_k, D_k)
     }
   \end{pcvstack}
   \pchspace
   \begin{pcvstack}
     \procedure{$\vpssver{0}(t, \mu, C, \{D_j\}_{j\in C})$}{
       \rule{0pt}{\baselineskip}
       \pcassert \verts{C}\geq \mu\\
       \pcfor j\in C \pcdo\\
       \t (D_{j, \aset})_{\aset\in\asetsof{j}} \defeq D_j\\
       \pcif \forall \aset\in \asets, \verts{\{D_{j,\aset}\}_{j\in C\setminus \aset}} = 1 \pcthen\\
       \t \pcreturn 1\\
       \pcelse \pcreturn 0
     }
     \pcvspace
     \procedure{$\agg(t, \mu, C, \{D_j\}_{j\in C})$}{
       \rule{0pt}{\baselineskip}
       \pcfor \aset\in \asets\pcdo\\
       \t \{D_{\aset}'\} \defeq \{D_{j,\aset}\}_{j\in C\setminus \aset}\\
       \pcreturn D\defeq \sprod_{\aset\in\asets}D_{\aset}'
     }
     \pcvspace
     \procedure{$\recover(t, \mu, C, \{d_j\}_{j\in C})$}{
       \rule{0pt}{\baselineskip}
       \pcassert \verts{C}\geq \mu\land C\subseteq \interval{n}\\
       \pcfor j\in C\pcdo\\
       \t (d_{j,\aset})_{\aset\in\asetsof{j}} \defeq d_j\\
       \t \pcfor \aset\in\asetsof{j}\pcdo\\
       \t \t D_{j, \aset}\defeq g^{d_{j,\aset}}\\
       \t D_j \defeq (D_{j,\aset})_{\aset\in\asetsof{j}}\\
       \pcassert \vpssver{0}(t, \mu, C, \{D_j\}_{j\in C})\\
       \pcfor \aset\in\asets\pcdo\\
       \t \{d_{\aset}'\} \defeq \{d_{j,\aset}\}_{j\in C\setminus \aset}\\
       d \defeq \ssum_{\aset\in\asets} d_{\aset}'\pcsc D\defeq g^d\\
       \pcreturn (d, D)
     }
   \end{pcvstack}
  \end{pchstack}
  \end{varwidth}}
  \caption{The protocol $\vpss{0}$, a variant of $\vpss{1}$ without share conversion. Notation for subsets is as in \cref{fig:vpss1}.}
  \label{fig:vpss0}
\pairfigmid

  \fitbox{\columnwidth}{%
  \begin{varwidth}{2\columnwidth}
  \begin{pchstack}[boxed,center]
   \pcsetargs{linenumbering}
   \begin{pcvstack}
     \procedure{$\setup(\secparam)$}{
       \rule{0pt}{\baselineskip}
       \param \defeq \gparam \gets \grgen(\secparam)\\
       \pcreturn \param
     }
     \pcvspace
     \procedure{$\keygen(n, t, \mu)$}{
       \rule{0pt}{\baselineskip}
       \pcassert \mu\geq 2t-1\\
       \asets \defeq \sbinom{\interval{n}}{t-1}\\
       \pcfor \aset\in\asets \pcdo\\
       \t \phi_{\aset} \sample \Zpp \pccomment{$\sk = \ssum_{\aset\in\asets}\phi_{\aset}$}\\
       \pcfor k\in \interval{n} \pcdo\\
       \t \sk_k \defeq \{(\aset, \phi_{\aset})\}_{\aset\in\asetsof{k}}\\
       \pcreturn (\sk_1,\ldots, \sk_n) \pccomment{$\verts{\sk_k} = \sbinom{n-1}{t-1}$}
     }
     \pcvspace
     \procedure{$\gen(k, \sk_k, w)$}{
       \rule{0pt}{\baselineskip}
       \{(\aset, \phi_{\aset})\}_{\aset\in\asetsof{k}} \defeq \sk_k\\
       \pcfor \aset\in\asetsof{k}\pcdo\\
       \t L'_{\aset}(x)\defeq \sprod_{j\in \aset}\frac{j-x}{j}\\
       \pclinecomment{the values $L'_{\aset}(k)$ can be precomputed}\\
       d_k \defeq \ssum_{\aset\in\asetsof{k}}\Hprf(\phi_{\aset}, w)\cdot L'_{\aset}(k)\\
       D_k \defeq g^{d_k}\\
       \pcreturn (d_k, D_k)
     }
   \end{pcvstack}
   \pchspace
   \begin{pcvstack}
     \procedure{$\vpssver{1}(t, \mu, C, \{D_j\}_{j\in C})$}{
       \rule{0pt}{\baselineskip}
       \pcassert \verts{C}\geq \mu\\
       \pcfor i\defeq t,\ldots, \verts{C}-1\pcdo\\
       \t \pcfor j\in C\pcdo\\
       \t \t L_j(x) \defeq \sprod_{k\in C\setminus\{j\}}\frac{x-k}{j-k}\\
       \t \t \lambda_{j, i} \defeq [x^i](L_j(x))\\
       \t B_i \defeq \sprod_{j\in C}D_j^{\lambda_{j,i}}\\
       \t \pcif B_i\neq 1_\GG\pcthen \pcreturn 0\\
       \pcreturn 1
     }
     \pcvspace
     \procedure{$\agg(t, \mu, C, \{D_j\}_{j\in C})$}{
       \rule{0pt}{\baselineskip}
       \pcfor j\in C\pcdo\\
       \t \lambda_j \defeq \sprod_{k\in C\setminus \{j\}}\frac{k}{k - j}\\
       \pcreturn D\defeq \sprod_{j\in C} D_j^{\lambda_j}
     }
     \pcvspace
     \procedure{$\recover(t, \mu, C, \{d_j\}_{j\in C})$}{
       \rule{0pt}{\baselineskip}
       \pcassert \verts{C}\geq \mu\land C\subseteq \interval{n}\\
       \pcfor j\in C\pcdo\\
       \t D_j\defeq g^{d_j}\\
       \t \lambda_j \defeq \sprod_{k\in C\setminus \{j\}}\frac{k}{k - j}\\
       \pcassert \vpssver{1}(t, \mu, C, \{D_j\}_{j\in C})\\
       d \defeq \ssum_{j\in C} d_j\cdot \lambda_j\pcsc D \defeq g^d\\
       \pcreturn (d, D)
     }
   \end{pcvstack}
  \end{pchstack}
  \end{varwidth}}
  \caption{The protocol $\vpss{1}$ from Arctic~\cite{KG24}. We write $\asets = \sbinom{\interval{n}}{t-1}$ for the family of all $(t-1)$-subsets of $\interval{n}$ and $\asetsof{k} = \{\aset\in\asets : k\notin\aset\}$ for the subsets whose shares member $k$ holds.}
  \label{fig:vpss1}
\pairfigend

\pairfigstart

  \fitbox{\columnwidth}{%
  \begin{varwidth}{2\columnwidth}
  \begin{pcvstack}[boxed,center]
   \pcsetargs{linenumbering}
   \begin{pchstack}
    \begin{pcvstack}
      \procedure{$\setup(\secparam)$}{
        \rule{0pt}{\baselineskip}
        \param \defeq \gparam \gets \grgen(\secparam)\\
        \pcreturn \param
      }
      \pcvspace
      \procedure{$\preround(k, \sk_k, \sid)$}{
        \rule{0pt}{\baselineskip}
        \pclinecomment{$\sid$ is assumed to have a fixed length}\\
        \pcfor i\in\interval{\nu}\pcdo\\
        \t (r_{i,k}, R_{i,k}) \defeq \highlight{\vpss{0}.\gen}(k, \sk_k, i\concat \sid)\\
        \pcreturn (R_{1,k},\ldots, R_{\nu, k})
      }
    \end{pcvstack}
    \pchspace
    \begin{pcvstack}
      \procedure{$\keygen(k, n, t, A)$}{
        \rule{0pt}{\baselineskip}
        \pcassert n\geq 2t-1\\
        \sk_k \gets \highlight{\vpss{0}.\keygen}(n, t, 2t-1)[k]\\
        (x_k, X_k)\defeq \highlight{\vpss{0}.\gen}(k, \sk_k, w_0)\\
        (C, \{X_j\}_{j\in C}) \gets A(X_k)\\
        \pcassert \highlight{\vpssver{0}}(t, 2t-1, C, \{X_j\}_{j\in C}) = 1\\
        X \defeq \highlight{\vpss{0}.\agg}(t, 2t-1, C, \{X_j\}_{j\in C})\\
        \pcreturn (X, \sk_k)
      }
    \end{pcvstack}
   \end{pchstack}
   \pcvspace
   \begin{pchstack}
    \begin{pcvstack}
      \procedure{$\preagg(\pk, C, (R_{i,j})_{i\in\interval{\nu},j\in C})$}{
        \rule{0pt}{\baselineskip}
        \pcfor i\in\interval{\nu}\pcdo\\
        \t \pcassert \highlight{\vpssver{0}}(t, 2t-1, C, \{R_{i,j}\}_{j\in C}) = 1\\
        \t R_i' \defeq \highlight{\vpss{0}.\agg}(t, 2t-1, C, \{R_{i,j}\}_{j\in C})\\
        b_1 \defeq \Hnon(\pk, (R_1',\ldots, R_\nu'))\\
        \pcfor i \in \interval{\nu}\pcdo\\
        \t R_i \defeq (R_i')^{b_1^{i-1}}\\
        \pcreturn (R_1, \ldots, R_\nu)
      }
    \end{pcvstack}
    \pchspace
    \begin{pcvstack}
      \procedure{$\signagg(R, C, \{s_k\}_{k\in C})$}{
        \rule{0pt}{\baselineskip}
        \pclinecomment{$R$ can be computed from protocol}\\
        \pclinecomment{messages as in $\signround$}\\
        \pcfor k\in C\pcdo\\
        \t \highlight{(s_{k,\aset})_{\aset\in\asetsof{k}} \defeq s_k}\\
        \highlight{\pcfor \aset\in\asets\pcdo}\\
        \t \highlight{\{s_{\aset}'\} \defeq \{s_{k,\aset}\}_{k\in C\setminus\aset}}\\
        \highlight{s \defeq \ssum_{\aset\in \asets} s_{\aset}'}\\
        \pcreturn (R, s)
      }
    \end{pcvstack}
   \end{pchstack}
   \pcvspace
   \procedure{$\signround(k, \sk_k, \pk, \sid, m, (R_i')_{i\in\interval{\nu}}, \overline{j}, \{\overline{\pk}_i\}_{i\in\interval{\overline{n}}\setminus \{\overline{j}\}}, (\overline{R}_i)_{i\in\interval{\nu}})$}{
     \rule{0pt}{\baselineskip}
     \highlight{(x_{k,\aset}, X_{k,\aset})_{\aset\in\asetsof{k}} \defeq \vpss{0}.\gen(k, \sk_k, w_0)}\\
     \pcfor i\in\interval{\nu}\pcdo\\
     \t (r_{i,k}, R_{i,k}) \defeq \highlight{\vpss{0}.\gen}(k, \sk_k, i\concat \sid)\\
     \t \highlight{(r_{i,k,\aset})_{\aset\in\asetsof{k}} \defeq r_{i,k}}\\
     \overline{\pk}_{\overline{j}} \defeq \pk\pcsc L\defeq (\overline{\pk}_j)_{j\in\interval{\overline{n}}}\\
     \tX\defeq \keyagg(L) \pcsc a \defeq \keyaggcoef(L, \pk)\\
     b_1 \defeq \Hnon(\pk, (R_1',\ldots, R_\nu'))\pcsc b_0 \defeq \Hnonover(\tX, (\overline{R}_1,\ldots, \overline{R}_\nu), m)\\
     \check{b} \defeq b_1\cdot b_0\pcsc R \defeq \sprod_{i\in\interval{\nu}}\overline{R}_i^{b_0^{i-1}}\pcsc c \defeq \Hsig(\tX, R, m)\\
     \highlight{\pcfor \aset\in\asetsof{k}\pcdo}\\
     \t \highlight{s_{k,\aset} \defeq \ssum_{i\in\interval{\nu}}r_{i,k,\aset}\,\check{b}^{i-1} + c\cdot a\cdot x_{k,\aset}}\\
     \pcreturn \highlight{s_k \defeq (s_{k,\aset})_{\aset\in\asetsof{k}}}
   }
  \end{pcvstack}
  \end{varwidth}}
  \caption{The protocol $\scheme{0}[\grgen, \nu, t]$. Changes compared to $\scheme{1}[\grgen, \nu, t]$ are \highlight{\text{highlighted}}. Notation for subsets is as in \cref{fig:vpss1}.}
  \label{fig:iceberg0}
\pairfigmid

  \fitbox{\fullscale\columnwidth}{%
  \begin{varwidth}{2\columnwidth}
  \begin{pchstack}[boxed,center]
   \pcsetargs{linenumbering}
   \begin{pcvstack}
    \procedure{$\setup(\secparam)$}{%
      \rule{0pt}{\baselineskip}
      \gparam \gets \grgen(\secparam) \\
      \text{Select four hash functions}\\
      \Hagg, \Hnon, \Hnonover, \Hsig:\str\to\Zpp\\
      \param \defeq\\
      \t (\gparam, \Hagg, \Hnon, \Hnonover, \Hsig)\\
      \pcreturn \param
    }
    \pcvspace
    \procedure{$\keygen()$}{%
      \rule{0pt}{\baselineskip}
      x \sample \ZZ_p \pcsc X \defeq g^x \\
      \sk \defeq x \pcsc \pk \defeq X \\
      \pcreturn (\sk,\pk)
    }
    \pcvspace
    \procedure{$\musigcoef(L,X_i)$}{%
      \rule{0pt}{\baselineskip}
      \pcreturn \Hagg(L,X_i)
    }
    \pcvspace
    \procedure{$\keyagg(L)$}{%
      \rule{0pt}{\baselineskip}
      \{X_1,\ldots,X_n\} \defeq L \\
      \pcfor i \defeq 1 \ldots n \pcdo \\
      \t a_i \defeq \musigcoef(L,X_i)\\
      \pcreturn \tX \defeq \sprod_{i=1}^n X_i^{a_i}
    }
    \pcvspace
    \procedure{$\ver(\tX,m,\sigma)$}{%
      \rule{0pt}{\baselineskip}
      (R,s) \defeq \sigma \\
      c \defeq \Hsig(\tX,R,m) \\
      \pcreturn (g^s = R\tX^c)
    }
    \pcvspace
    \procedure{$\sign()$}{%
      \rule{0pt}{\baselineskip}
      \pclinecomment{Local signer has index $1$.} \\
      \pcfor j \defeq 1 \ldots \nu \pcdo \\
        \t r_{1,j} \sample \ZZ_p \pcsc R_{1,j} \defeq g^{r_{1,j}} \\
      \msg_1 \defeq (R_{1,1},\ldots,R_{1,\nu}) \\
      \state_1 \defeq (r_{1,1},\ldots,r_{1,\nu}) \\
      \pcreturn (\msg_1,\state_1)
    }
   \end{pcvstack}
   \pchspace[2em]
   \begin{pcvstack}
    \procedure{$\signagg(\msg_1,\ldots,\msg_n)$}{%
     \rule{0pt}{\baselineskip}
     \pcfor i \defeq 1 \ldots n \pcdo \\
     \t (R_{i,1}, \ldots, R_{i,\nu}) \defeq \msg_i\\
     \pcfor j \defeq 1 \ldots \nu \pcdo \\
     \t R_j \defeq \sprod_{i=1}^n R_{i,j}\\
     \pcreturn (R_1, \ldots, R_\nu)
    }
    \pcvspace
    \procedure{$\signaggext(\msg, \tX)$}{%
     \rule{0pt}{\baselineskip}
     (R_1', \ldots, R_\nu') \defeq \msg\\
     b \defeq \Hnon(\tX, (R_1', \ldots, R_\nu'))\\
     \pcfor j \defeq 1 \ldots \nu \pcdo \\
     \t R_j \defeq (R_j')^{b^{j-1}}\\
     \pcreturn (R_1, \ldots, R_\nu)
    }
    \pcvspace
    \procedure{$\sign'(\state_1,(\msg^d)_{0\leq d<\Lambda},\sk_1,m,\{\pk_{i,d}\}_{2\leq i\leq n_d,0\leq d < \Lambda})$}{%
      \rule{0pt}{\baselineskip}
      \pclinecomment{$\sign'$ must be called at most once per $\state_1$.} \\
      (r_{1,1},\ldots,r_{1,\nu}) \defeq \state_1 \pcsc \pk_{1,\Lambda-1} \defeq g^{\sk_1} \\
      L_{\Lambda - 1} \defeq \{\pk_{1,\Lambda-1},\ldots,\pk_{n_{\Lambda-1},\Lambda-1}\} \\
      a_{1,\Lambda-1} \defeq \musigcoef(L_{\Lambda-1},\pk_{1,\Lambda-1})\\
      \pk_{1,\Lambda-2} \defeq \keyagg(L_{\Lambda-1})\\
      \pcfor d\defeq \Lambda - 2,\ldots,0\pcdo\\
      \t b_{d+1} \defeq \Hnon(\pk_{1,d}, \msg^{d+1})\\
      \t L_d \defeq \{\pk_{1,d},\ldots, \pk_{n_d, d}\}\\
      \t a_{1,d} \defeq \musigcoef(L_d, \pk_{1,d})\\
      \t \pk_{1,d-1} \defeq \keyagg(L_d)\\
      \tX \defeq \pk_{1,-1}\pcsc (R_{1,0}, \ldots, R_{\nu,0}) \defeq \msg^0\\
      b_0 \defeq \Hnonover(\tX,(R_{1,0},\ldots,R_{\nu,0}),m) \\
      R \defeq \sprod_{j=1}^{\nu} R_{j,0}^{b_0^{j-1}} \pcsc \check{b} \defeq \sprod_{\ell=0}^{\Lambda-1}b_\ell\\
      c \defeq \Hsig(\tX,R,m) \pcsc \check{c} \defeq c\sprod_{\ell=0}^{\Lambda-1}a_{1,\ell}\\
      s_1 \defeq \check{c}\sk_1 + \ssum_{j=1}^{\nu} r_{1,j} \check{b}^{j-1} \bmod p \\
      \state'_1 \defeq R \pcsc \msg'_1 \defeq s_1 \\
      \pcreturn (\state'_1, \msg'_1)
    }
    \pcvspace
    \procedure{$\signagg'(\msg'_1,\ldots,\msg'_n, \state'_1)$}{%
      \rule{0pt}{\baselineskip}
      R \defeq \state'_1\\
      (s_1, \ldots, s_n) \defeq (\msg'_1, \ldots, \msg'_n)\\
      s \defeq \ssum_{i=1}^n s_i \bmod p \\
      \pcreturn \sigma\defeq (R, s)
    }
   \end{pcvstack}
  \end{pchstack}
  \end{varwidth}}
  \caption{The nested multi-signature scheme $\nestedmusig[\grgen,\nu]$~\cite{Koh26}. Public parameters $\param$ returned by $\setup$ are implicitly given as input to all other algorithms.}
  \label{fig:nestedmusig}
\pairfigend

\end{document}